\documentclass[12pt,letterpaper]{article}

\usepackage[margin=1in]{geometry}
\usepackage{setspace}
\usepackage{amsmath,amsthm,amssymb,mathtools}

\usepackage{booktabs}
\usepackage{array}
\usepackage{multirow}
\usepackage{tabularx}
\usepackage{longtable}
\usepackage{placeins}

\usepackage{graphicx}
\usepackage{tikz}
\usetikzlibrary{shapes.geometric, arrows.meta, positioning, calc}

\usepackage{enumitem}

\usepackage{xcolor}

\usepackage{algorithm}
\usepackage{algpseudocode}

\usepackage{natbib}
\usepackage{hyperref}
\hypersetup{
  colorlinks=true,
  linkcolor=black,
  citecolor=black,
  urlcolor=black
}
\usepackage[capitalise,noabbrev]{cleveref}

\newcommand{\E}{\mathbb{E}}
\newcommand{\Var}{\mathrm{Var}}
\newcommand{\Cov}{\mathrm{Cov}}
\newcommand{\pS}{P_{\mathcal{S}}}
\newcommand{\pT}{P_{\mathcal{T}}}
\newcommand{\muS}{\mu_{\mathcal{S}}}
\newcommand{\muT}{\mu_{\mathcal{T}}}
\newcommand{\tauDR}{\hat{\tau}_{\mathrm{DR}}}
\newcommand{\tauT}{\tau_{\mathcal{T}}}
\newcommand{\tauS}{\tau_{\mathcal{S}}}
\newcommand{\ESS}{\mathrm{ESS}}

\newcommand{\cS}{\mathcal{S}}
\newcommand{\cT}{\mathcal{T}}
\newcommand{\cX}{\mathcal{X}}

\newcommand{\indep}{\perp\!\!\!\perp}
\newcommand{\Op}{O_p}
\newcommand{\Sx}{S_X}
\newcommand{\what}[1]{\widehat{#1}}

\DeclareMathOperator{\sign}{sign}
\DeclareMathOperator{\SpearmanCorr}{SpearmanCorr}

\newtheorem{theorem}{Theorem}[section]
\newtheorem{proposition}[theorem]{Proposition}

\theoremstyle{definition}
\newtheorem{definition}[theorem]{Definition}
\newtheorem{assumption}{Assumption}

\theoremstyle{remark}
\newtheorem{remark}[theorem]{Remark}

\begin{document}

\title{When Can You Trust Your Synthetic Users?\\
  Diagnostics and Corrections for LLM Consumer Panels}

\author{Robson Tigre\thanks{Recast, email: \href{mailto:robson.tigre0@gmail.com}{robson.tigre0@gmail.com}, website: \url{https://www.robsontigre.com}} \and Hugo Gobato Souto\thanks{Dell Technologies}}

\date{}

\maketitle

\begin{abstract}
\small\singlespacing\noindent
Large language models are increasingly used to simulate consumers, yet agreement with aggregate survey outcomes can conceal errors in the subgroup estimates that guide decisions. We develop a framework for deciding when to trust, correct, or abandon a synthetic consumer panel by distinguishing population mismatch from differences in conditional responses. The framework combines diagnostics for covariate overlap, conditional calibration, and stability across models with an augmented inverse probability weighting estimator that treats LLM predictions as auxiliary information. Correction requires a complete invitation frame, response indicators for everyone invited, and human outcomes for respondents, together with assumptions of missing at random and positivity. We evaluate the framework using simulations and controlled nonresponse experiments constructed from paired human and LLM responses in ANES and Twin-2K-500. In identified simulation regimes, point bias remains small when either the response model or the outcome model is correct, but nominal 95\% confidence intervals cover only 88--94\% in some regimes. When nonresponse depends on missing human outcomes, substantial bias persists and coverage approaches zero. Controlled masking of the paired datasets reproduces this contrast between successful point correction under imposed missing at random and failure under nonignorable nonresponse. These findings show why validation must address the intended population and task, and why the scope for correcting synthetic estimates depends on the available human data and the assumptions that identify the target.

\vspace{6pt}
\end{abstract}

\section{Introduction}
\label{sec:introduction}

\begingroup
LLM-generated synthetic users promise faster iteration and greater scale than conventional consumer research. Recent studies evaluate purchase intent, pricing, and political attitudes \citep{maier2025ssr,brand2025llmsmarketresearch,bisbee2024synthetic}. But aggregate agreement is not the decision target. A panel can reproduce a topline while reversing subgroup patterns, compressing within-profile variation, or failing on a task that requires different behavior. The operational question is panel- and estimand-specific: which population is represented, which response function is being approximated, and what evidence supports using the result?

\Cref{tab:evidence_failures} summarizes why that distinction matters. The documented failures span political attitudes \citep{bisbee2024synthetic,argyle2023outofone}, consumer willingness-to-pay \citep{brand2025llmsmarketresearch}, purchase-intent subgroup analysis \citep{maier2025ssr}, and cross-national surveys \citep{taday2026persona}. They include weak within-profile variation, coefficient sign flips, subgroup breakdown, willingness-to-pay overstatement, and corrections that improve an aggregate target while worsening demographic subgroups. For example, \citet{krsteski2025validsurvey} show that population-level PPI++ rectification  increases female bias by 58.5\% while decreasing Non-Hispanic White bias by 44.7\%.
\endgroup

\begin{table}[htbp]
\centering
\caption{Evidence of systematic failures in LLM synthetic panels.}
\label{tab:evidence_failures}
\small
\begin{tabularx}{\textwidth}{>{\raggedright\arraybackslash}p{3cm} >{\raggedright\arraybackslash}p{3.5cm} X}
\toprule
\textbf{Evidence} & \textbf{Source} & \textbf{Key Finding} \\
\midrule

 Within-profile compression + coefficient divergence & {\citet{bisbee2024synthetic}} & The synthetic data exhibit substantially smaller within-profile variation; 48\% of coefficients differ significantly, and among those cases, 32\% flip sign \\[4pt]
Subgroup collapse & \citet{argyle2023outofone} & $\rho = 0.90$--$0.94$ aggregate but $\rho \approx 0.02$ for political Independents (2020 wave) \\[4pt]

 WTP overestimation & {\citet{brand2025llmsmarketresearch}} & For one built-in laptop-projector feature, baseline GPT-3.5 overestimated WTP by approximately $3\times$; fine-tuning substantially reduced the difference, and GPT-4o estimates were similar to human estimates \\[4pt]

Aggregate attainment, subgroup caution & {\citet{maier2025ssr}} & SSR correlation attainment is 90.2\% for GPT-4o; gender, region, and ethnicity patterns are not consistently replicated, and the authors caution against interpreting subgroup analyses from synthetic panels \\[4pt]

 Subgroup and novel-item error & {\citet{morris2025limitssynthetic}} & Within ${\sim}4$pp for one frequently polled polarized item; 9.8--18.2pp topline differences on novel immigration items; subgroup mean absolute error above 10pp \\[4pt]

 Persona conditioning has heterogeneous effects & {\citet{taday2026persona}} & Near-zero aggregate improvement; effects concentrated in a small subset of items and low-support strata across 70,000+ instances \\[4pt]
Global correction worsens subgroups & \citet{krsteski2025validsurvey} & PPI++ at population level \emph{increases} female bias by $+58.5\%$ \\
\bottomrule
\end{tabularx}
\end{table}

\begingroup
A marketing team can therefore be wrong for two different reasons. The synthetic panel may represent the wrong customers (wrong WHO), or it may produce the wrong conditional responses for otherwise comparable customers (wrong WHAT). Aggregate correlation does not distinguish these failures, and their remedies differ. Reweighting can address the WHO component only when overlap is adequate; it does not repair WHAT. Outcome calibration requires observed human outcomes, and either failure can be severe enough to require Walk Away.

Even a correct diagnosis does not make correction available. Trust requires human calibration evidence for the target question. Correct additionally requires a complete invitation frame, response indicators for all invited units, respondent outcomes, credible MAR, and positivity. When those conditions fail, the framework returns correction unavailable rather than forcing a stabilized estimate.

Existing research supplies the components of this decision. LLM validation studies document successes and failures \citep{horton2023homosilicus,argyle2023outofone,maier2025ssr,bisbee2024synthetic}; distribution-shift methods separate covariate from concept shift \citep{cai2025disde,jin2025beyondreweighting}; missing-data and nonprobability-sampling methods provide established correction machinery when the response frame and identifying assumptions exist \citep{robins1994aipw,dahabreh2019generalizing,chen2020doubly,elliott2017inference}. This paper combines those pieces into an LLM-panel-specific, segment-level diagnosis-and-action system.
\endgroup

\subsection{Contributions}

\begingroup
We make three contributions. They concern the LLM-panel mapping and decision workflow, not a new estimator family:

\begin{enumerate}[label=\arabic*.]
  \item \textbf{An LLM-panel diagnosis of wrong WHO and wrong WHAT.} We adapt the DISDE decomposition \citep{cai2025disde} from model-performance gaps to target-mean bias, then connect covariate shift and concept shift to overlap, conditional-calibration, and cross-model diagnostics (\Cref{sec:decomposition,sec:diagnostics}). The diagnostics attach to a question family, model version, prompt configuration, target population, and time---not to an LLM in general.

  \item \textbf{A segment-level Trust / Correct / Walk Away decision system.} Diagnostics can identify prespecified failures and end in Walk Away, but they do not certify Trust. Trust requires held-out human calibration evidence. Correct requires that evidence plus a complete invitation frame, response indicators, respondent outcomes, credible MAR, positivity, and a declared target reading (\Cref{sec:decision}).

  \item \textbf{A frame-grounded correction module with explicit identification boundaries.} We map  established missing-data and nonprobability-sample IPW/AIPW machinery \citep{robins1994aipw,chen2020doubly,elliott2017inference} to  an invitation-frame setting where the stored LLM prediction is a fixed auxiliary variable rather than a human outcome. We separate finite-frame from superpopulation targets and one-model point robustness from confidence-interval validity (\Cref{sec:correction}). Frozen simulations and controlled masking evaluate both identified regimes and the failure cases---both nuisances wrong, MNAR, and positivity refusal---that determine when correction is unavailable (\Cref{sec:empirical}).
\end{enumerate}
\endgroup

\subsection{Preview of Empirical Results}

\begingroup
The empirical section has three distinct roles (\Cref{sec:empirical}).

First, a frozen simulation battery evaluates frame-AIPW under known truth. In identified regimes, absolute point bias is below 0.03 when one nuisance model is correct, but outcome-correct, propensity-wrong intervals cover only 88.15--94.15\%. Both-wrong cells fail, and outcome-dependent MNAR produces bias around 0.233 with 0.0--0.2\% coverage.

Second, the direct paired ANES and Twin-2K applications illustrate diagnostics only because neither dataset records invitations or nonrespondents. ANES reveals subgroup calibration failures. Twin-2K shows no tested full-sample pricing failure but falls outside the operational Trust range for targeted subgroups and behavioral tasks.

Third, controlled masking retains real covariates, human outcomes, and LLM predictions while simulating response. Under imposed MAR, frame-AIPW has small bias and near-nominal or conservative coverage. Outcome-dependent MNAR produces material bias and severe coverage failure, while the prespecified positivity gate withholds 99 of 12,000 attempts.

Together, these exercises support diagnosis first and correction only when the invitation frame and identifying assumptions exist.
\endgroup

\section{Related Work}
\label{sec:related_work}

 Three questions organize the relevant literature: whether LLMs can replace survey respondents, why synthetic data fail, and which failures can be corrected. We review each before comparing our framework with the closest alternatives.

\subsection{Can LLMs Replace Survey Respondents?}
\label{sec:rw_replace}

\begingroup
\citet{horton2023homosilicus} introduced the ``homo silicus'' concept---LLMs endowed with preferences and subjected to economic experiments---and  found that LLM responses qualitatively reproduced findings from dictator games, fairness and status-quo preferences, and a minimum-wage hiring experiment. \citet{argyle2023outofone} scaled this idea through ``silicon sampling'': conditioning GPT-3 on real demographic backstories from the American National Election Study to produce synthetic response distributions achieving tetrachoric correlations of $0.90$--$0.94$. In the consumer domain, \citet{brand2025llmsmarketresearch}  reported that, in one laptop study, GPT-3.5 estimates for storage and RAM were 64--90\% of human willingness to pay, and \citet{maier2025ssr} evaluated SSR across 57 consumer surveys and reported 90.2\% correlation attainment for GPT-4o.

 These studies illustrate commercial and methodological interest in synthetic respondents \citep{arora2025hybridaihumanmarketing,toubia2025digitaltwin}, with silicon sampling constituting 34\% of the LLM-survey papers in \citet{buskirk2025moreparameters}. We enter this conversation with a narrower question: what target-specific human evidence supports Trust, and when should the available data and assumptions instead route the analysis to Correct or Walk Away?
\endgroup

\subsection{What Goes Wrong with LLM Synthetic Data?}
\label{sec:rw_wrong}

A growing body of evidence documents systematic failures along two dimensions.

\begingroup
\paragraph{Covariate shift.} A synthetic panel exhibits covariate shift when its observed profile distribution differs from the target population. This difference can arise from how profiles are selected or generated and should be checked against the target frame. Separately, \citet{santurkar2023whoseopinions} find that human-feedback-tuned models align more closely with opinions from liberal, higher-income, well-educated, and nonreligious U.S. groups; the authors cannot recover the demographics of the crowdworkers who provided feedback. \citet{li2025personacatch} find that simulated opinions tend to drift farther from human responses as more persona content is generated by an LLM.
\endgroup

\paragraph{Concept shift.} Even when prompted with correct demographics, LLM conditional response functions diverge from real consumer behavior. \citet{bisbee2024synthetic}  document that 48\% of coefficients differ significantly and, among those cases, 32\% flip sign in regression models estimated from synthetic versus real ANES data. \citet{brand2025llmsmarketresearch}  report that, for one built-in laptop-projector feature, baseline GPT-3.5 overestimated WTP by approximately $3\times$; fine-tuning substantially reduced the difference, and GPT-4o estimates were similar to human estimates. \citet{kaiser2025aspire} report $d = 1.96$ over-positivity bias. \citet{taday2026persona}  find that persona conditioning across 70,000+ WVS-7 instances yields near-zero aggregate improvement, with heterogeneous effects concentrated in a small subset of items and low-support strata.  \citet{peng2025funhouse}  analyze LLM digital twins trained on 500 or more prior responses per person. They report an individual-level correlation of only $r = 0.20$ and identify five systematic distortions: stereotyping, insufficient individuation, representation bias, ideological biases, and hyper-rationality.

\paragraph{Interaction effects.}  Across 680 replication studies, \citet{jin2025beyondreweighting}   use a standardized pivotal-covariate measure to place an empirical bound on the conditional shift that could accompany the observed covariate shift. This motivates diagnosing both forms of shift rather than assuming that reweighting will be sufficient. All documented failures are special cases of the covariate shift + concept shift decomposition we formalize in \Cref{sec:decomposition}.

\subsection{How Should We Correct LLM Synthetic Data?}
\label{sec:rw_correct}

Statistical correction methods fall into five categories with progressive limitations:

\begin{enumerate}[label=(\roman*)]
  \item \textbf{Prompt engineering} (persona conditioning, SSR, RAG, fine-tuning): Improves prediction quality but provides no valid inference guarantees \citep{cho2024doppelganger, maier2025ssr, cao2025specializing}.

  \item \textbf{Importance reweighting}:   Corrects differences in observed covariate distributions under overlap {\citep{shimodaira2000covariateshift,hainmueller2012entropy}}. In our decomposition, however, reweighting remains biased when the conditional response function differs across source and target populations.

  \item \textbf{Prediction-powered inference (PPI)}: Corrects concept shift via a rectifier estimated on labeled data. \citet{angelopoulos2023ppi,angelopoulos2023ppiplusplus} introduce the framework; \citet{wang2025finetuningrectification} and \citet{krsteski2025validsurvey} apply it to LLM surveys. However, all PPI methods \emph{require} $\pS(X) = \pT(X)$---an assumption explicitly stated in proofs \citep[p.~7]{wang2025finetuningrectification} and maintained throughout all evaluations \citep[Appendix B.1]{krsteski2025validsurvey}.

  \item \textbf{Distribution shift diagnostics}: \citet{cai2025disde} provide the DISDE decomposition; \citet{jin2025beyondreweighting} bound concept shift from covariate shift; \citet{zhang2025datausefulness} introduce the Data Usefulness Coefficient; \citet{huang2025howmany} offer uncertainty quantification.  These methods diagnose, decompose, or quantify distribution shift. {\citet{jin2025beyondreweighting}} also provide a covariate-adjusted point estimate and prediction intervals, but these methods do not by themselves supply this paper's frame-AIPW correction.

  \item \textbf{Doubly robust estimation}: \citet{robins1994aipw} introduce AIPW; \citet{dahabreh2019generalizing} apply it to transportability; \citet{chernozhukov2018dml} provide cross-fitting for honest ML-based inference. Our correction module inherits directly from the nonprobability sampling literature: \citet{chen2020doubly} establish DR estimation for combining nonprobability and probability samples, and \citet{elliott2017inference} provide the statistical framework for inference from such samples.  LLM synthetic panels are generated without a probability-based selection mechanism, making them extreme nonprobability samples. To our knowledge, prior work has not applied this framework to LLM synthetic populations.
\end{enumerate}

We argue that diagnostics must precede correction: \citet{wang2025finetuningrectification}, \citet{krsteski2025validsurvey}, and \citet{zou2026generalizedppi} all jump to correction without diagnosing the nature of the shift, while \citet{guerdan2025drjudge} provide a DR estimator but assume no concept drift.

\subsection{Competitive Positioning}
\label{sec:rw_positioning}

\Cref{tab:competitive} maps the competitive landscape. The four-way intersection---concept shift, covariate shift, DR guarantees, and diagnostic framework---remains unoccupied.

\begin{table}[htbp]
\centering
\caption{Competitive positioning. Each column indicates whether the paper addresses the corresponding dimension. Only our paper occupies the four-way intersection.}
\label{tab:competitive}
\small
\begin{tabularx}{\textwidth}{>{\raggedright\arraybackslash}p{3.2cm} c c c c c >{\raggedright\arraybackslash}X}
\toprule
\textbf{Paper} & \rotatebox{55}{\textbf{Covariate}} & \rotatebox{55}{\textbf{Concept}} & \rotatebox{55}{\textbf{DR}} & \rotatebox{55}{\textbf{Diagnostics}} & \rotatebox{55}{\textbf{Consumer}} & \textbf{Key Limitation} \\
\midrule
\textbf{This paper} & \checkmark & \checkmark & \checkmark & \checkmark & \checkmark & --- \\[3pt]
Guerdan et al.\ (2025) & \checkmark & --- & \checkmark & --- & --- & Assumes no concept drift \\[3pt]
Wang et al.\ (2025) & --- & Partial & --- & --- & --- & Same-population assumption \\[3pt]
Krsteski et al.\ (2025) & --- & Partial & --- & --- & --- & Same-population assumption \\[3pt]
Zou et al.\ (2026) & \checkmark & --- & --- & --- & --- & Assumes $C \indep Y | X$ \\[3pt]
\citet{huang2025howmany} & \checkmark & \checkmark & --- & Partial & --- & No corrected estimates \\[3pt]
\citet{datta2025ppiipw} & --- & --- & --- & --- & --- & Same population \\[3pt]
Song et al.\ (2026) & --- & --- & Refs & Partial & --- & Tutorial, no implementation \\[3pt]
Cai et al.\ (2025) & \checkmark & \checkmark & --- & \checkmark & --- & ML performance, not estimation \\
\bottomrule
\end{tabularx}
\end{table}

\section{Bias Decomposition}
\label{sec:decomposition}

We now formalize the sources of error when using an LLM-generated panel to estimate a target population quantity. The decomposition adapts the DISDE framework \citep{cai2025disde} from diagnosing ML model performance degradation to diagnosing estimation bias in synthetic consumer panels. The adaptation is non-trivial: we re-interpret the ``loss'' as estimation error rather than prediction error, the ``training distribution'' as the LLM source population, and the ``target distribution'' as the firm's customer base.

\subsection{Setup and Notation}
\label{sec:setup}

\begin{definition}[Source and Target Populations]
\label{def:populations}
Let $\cS$ denote the \emph{source population}---the implicit population from which the LLM generates synthetic consumer profiles---with joint distribution $\pS(X, Y)$ over demographics $X \in \cX \subseteq \mathbb{R}^d$ and outcomes $Y \in \mathbb{R}$. Let $\cT$ denote the \emph{target population}---the firm's actual customer base---with joint distribution $\pT(X, Y)$.
\end{definition}

The researcher observes a large synthetic sample $\{(X_i^{\cS}, Y_i^{\cS})\}_{i=1}^N$ from $\pS$ (generated by querying the LLM) and a small real sample $\{(X_j^{\cT}, Y_j^{\cT})\}_{j=1}^n$ from $\pT$ (collected from actual customers), where typically $N \gg n$. Define the conditional response functions:
\begin{equation}
  \muS(x) \coloneqq \E_{\cS}[Y \mid X = x], \qquad
  \muT(x) \coloneqq \E_{\cT}[Y \mid X = x],
  \label{eq:conditional_means}
\end{equation}
and the target estimand:
\begin{equation}
  \tauT \coloneqq \E_{\cT}[Y] = \int \muT(x) \, p_{\cT}(x) \, dx.
  \label{eq:target_estimand}
\end{equation}

The naive synthetic estimator is $\hat{\tau}_{\cS} = N^{-1} \sum_{i=1}^N Y_i^{\cS}$, which converges to $\tauS = \E_{\cS}[Y] = \int \muS(x) \, p_{\cS}(x) \, dx$. The bias of this estimator is $\tauT - \tauS$.

\begingroup
We connect this setting to the external validity framework of \citet{egami2023externalvalidity}. Their X-validity concerns generalizing across populations and requires conditional exchangeability together with overlap or positivity. Their Y-validity concerns substituting outcome measures within a population; it does not correspond to our concept-shift condition. Our framework therefore uses related ideas but defines its own covariate- and concept-shift conditions.
\endgroup

\subsection{Informal Decomposition}
\label{sec:informal_decomposition}

The bias of the naive synthetic estimator admits a simple additive decomposition:
\begin{align}
  \tauT - \tauS &= \int \muT(x) \, p_{\cT}(x) \, dx - \int \muS(x) \, p_{\cS}(x) \, dx \notag \\
  &= \underbrace{\int \bigl[p_{\cT}(x) - p_{\cS}(x)\bigr] \muT(x) \, dx}_{\text{Covariate shift bias}} + \underbrace{\int p_{\cS}(x) \bigl[\muT(x) - \muS(x)\bigr] dx}_{\text{Concept shift bias}} + \underbrace{\Op(N^{-1/2})}_{\text{Noise}}.
  \label{eq:informal_decomposition}
\end{align}

\begin{itemize}[leftmargin=2em]
  \item \textbf{Term 1 (Covariate shift):} $p_{\cS}(X) \neq p_{\cT}(X)$.      A synthetic panel exhibits covariate shift when its observed profile distribution differs from the target population. This difference can arise from how profiles are selected or generated and should be checked against the target frame. \emph{Potentially fixable via reweighting when overlap holds.}

  \item \textbf{Term 2 (Concept shift):} $\muS(x) \neq \muT(x)$. Even when prompted with correct demographics, the LLM's conditional response function differs from real consumer behavior. \citet{bisbee2024synthetic}  document that 48\% of coefficients differ significantly and, among those cases, 32\% flip sign; \citet{brand2025llmsmarketresearch}  report that, for one built-in laptop-projector feature, baseline GPT-3.5 overestimated WTP by approximately $3\times$; fine-tuning substantially reduced the difference, and GPT-4o estimates were similar to human estimates. \emph{Not fixable by reweighting alone---requires outcome model calibration.}

  \item \textbf{Term 3 (Noise):} Sampling variability and stochastic LLM output variation (temperature, prompt sensitivity). \emph{Vanishes at rate $\Op(N^{-1/2})$.}
\end{itemize}

\noindent The population-level decomposition (Online Appendix~J) is exact; the third term in \eqref{eq:informal_decomposition} reflects finite-sample estimation noise that arises when moving from the population identity to its sample analog.

\subsection{Formal Decomposition and Assumptions}
\label{sec:formal_decomposition}

The informal decomposition in \eqref{eq:informal_decomposition} can be made rigorous via the DISDE telescoping sum \citep[Equation 2.1]{cai2025disde}, which introduces a shared intermediate distribution $p_{\Sx}(x) \propto p_{\cS}(x) p_{\cT}(x) / [p_{\cS}(x) + p_{\cT}(x)]$ to cleanly separate covariate and concept shift. The decomposition requires three standard assumptions: (i)~shared support ($\mathrm{supp}(\pT) \subseteq \mathrm{supp}(\pS)$, testable via Diagnostic~1), (ii)~bounded density ratios, and (iii)~finite second moments. Online Appendix~J provides the full formal decomposition, assumption statements, and estimation details.

The key takeaway is that concept shift---the middle term, measuring how differently the LLM responds compared to real consumers \emph{at the same demographics}---cannot be fixed by reweighting alone. When concept shift dominates, the practitioner needs outcome model calibration (our DR correction in \Cref{sec:correction}) or must walk away.

\begingroup
For a direct-to-consumer company forecasting demand for a new household product, the distinction determines the remedy. If the synthetic panel contains too many young urban shoppers, reweighting can correct who counts in the forecast. If the LLM also overstates purchase intent among young urban shoppers relative to comparable real customers, changing the weights leaves that response error intact; the company needs outcome calibration or should not use the synthetic estimate.
\endgroup

\section{Diagnostic Toolkit}
\label{sec:diagnostics}

We provide three core diagnostics in the main text and two supplementary diagnostics in Online Appendix~B. Each diagnostic returns an interpretable metric with decision thresholds summarized in \Cref{tab:thresholds}.

\subsection{Diagnostic 1: Covariate Overlap}
\label{sec:diag_overlap}

\paragraph{Question.} Are the firm's customer profiles represented in the synthetic panel?

\paragraph{Estimation.} Train a binary classifier to distinguish LLM-generated profiles from real customer profiles using observed demographics $X$. Compute the propensity score $e(x) = \what{P}(\text{source} = \text{LLM} \mid X = x)$. Convert to density ratio weights:
\begin{equation}
  w_i = \frac{1 - e(X_i^{\cS})}{e(X_i^{\cS})},
  \label{eq:density_ratio}
\end{equation}
\noindent In practice, weights are self-normalized: $\tilde{w}_i = w_i / \sum_{i'} w_{i'}$, avoiding dependence on $N$ and $n$.

Calculate the effective sample size:
\begin{equation}
  \ESS = \frac{\bigl(\sum_{i=1}^N w_i\bigr)^2}{\sum_{i=1}^N w_i^2}.
  \label{eq:ess}
\end{equation}

\paragraph{Decision thresholds.}
\begin{itemize}[leftmargin=2em]
  \item $\ESS / N > 0.10$: \textbf{Adequate overlap.} Proceed to conditional calibration.
  \item $\ESS / N \in [0.01, 0.10]$: \textbf{Marginal overlap.} Reweighting is possible but high-variance. Proceed with caution.
  \item $\ESS / N < 0.01$: \textbf{Poor overlap. Walk Away.} The LLM's implicit population does not cover the firm's customer profiles. Reweighting produces extreme weights and unstable estimates.
\end{itemize}

\paragraph{Interpretation.} The ESS measures how many ``effective'' observations remain after reweighting the synthetic panel to match the target demographics. An ESS of 100 from a panel of 10,000 ($\ESS/N = 0.01$) means that most of the information concentrates on a few observations with extreme weights---the reweighted estimate is effectively based on 100 observations, not 10,000. The density ratio weights $w_i$ from \eqref{eq:density_ratio} are the same weights used in the DR estimator (\Cref{sec:correction}), so a low ESS directly implies high variance in the corrected estimate.

\paragraph{AUC failsafe.} ESS becomes unreliable at extremes of distributional separation: near-zero overlap can paradoxically inflate ESS. We supplement with the domain classifier AUC, which increases monotonically with separation. If AUC~$> 0.95$, the diagnostic triggers walk away regardless of ESS; if AUC~$< 0.90$ but ESS/N~$< 0.01$, low ESS reflects weight concentration rather than genuine non-overlap. \Cref{fig:overlap_scenarios} illustrates the three regimes.

\begin{figure}[htbp]
\centering
\includegraphics[width=\textwidth]{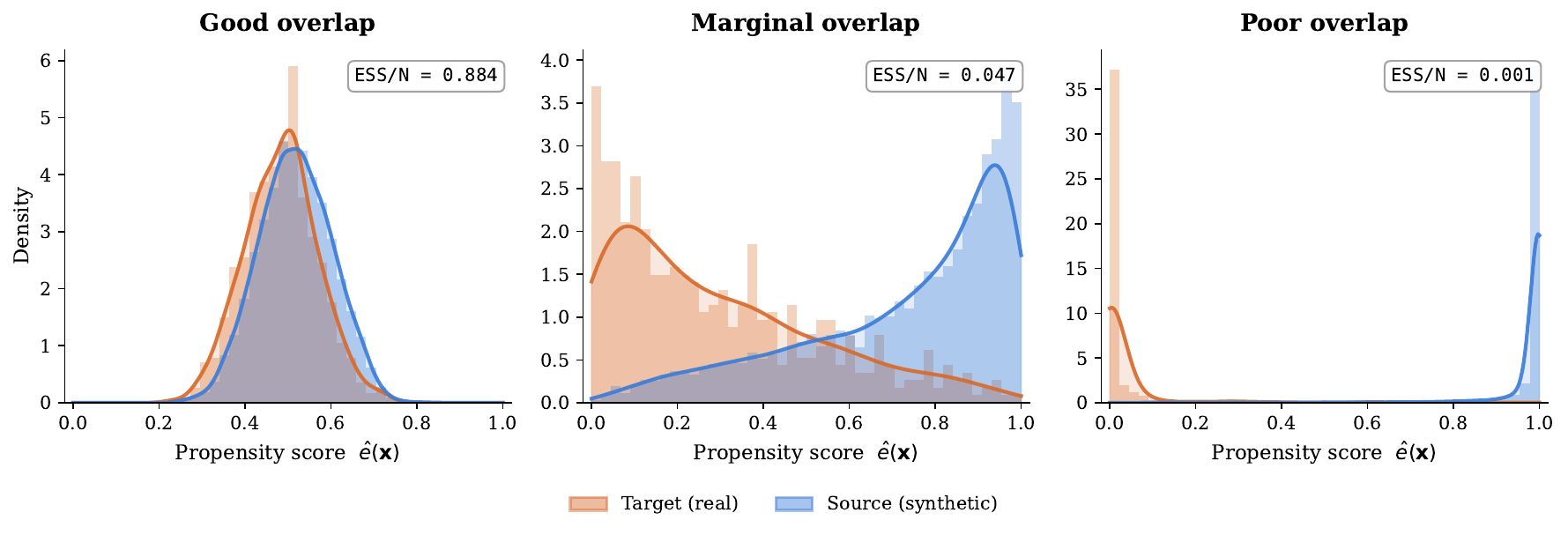}
\caption{Propensity score distributions under three covariate shift regimes from the controlled DGP. \textbf{Left:} Good overlap ($\ESS/N = 0.884$)---source and target populations are nearly indistinguishable, reweighting is reliable. \textbf{Center:} Marginal overlap ($\ESS/N = 0.047$)---moderate separation, reweighting is feasible but high-variance. \textbf{Right:} Poor overlap ($\ESS/N = 0.001$)---severe separation, nearly all importance weight concentrates on a few observations, making reweighted estimates unreliable.}
\label{fig:overlap_scenarios}
\end{figure}

\subsection{Diagnostic 2: Conditional Calibration}
\label{sec:diag_calibration}

\paragraph{Question.} Conditional on demographics, do synthetic responses track real responses?

\paragraph{Estimation.} Using the small real sample, compute subgroup-level calibration slopes. For each major demographic subgroup $g \in \mathcal{G}$ (e.g., age brackets, gender, income quartiles):
\begin{equation}
  \beta_g = \frac{\Cov(Y_{\text{real}}, Y_{\text{synthetic}} \mid \text{group} = g)}{\Var(Y_{\text{synthetic}} \mid \text{group} = g)}.
  \label{eq:calibration_slope}
\end{equation}

A perfectly calibrated synthetic panel yields $\beta_g = 1$ for all $g$. Systematic deviations indicate concept shift.

\paragraph{Decision thresholds.}
\begin{itemize}[leftmargin=2em]
  \item $\beta_g \in [0.5, 1.5]$ for all major subgroups and same sign as $\beta_{\text{overall}}$: \textbf{Moderate concept shift. Correct using DR} (\Cref{sec:correction}).
  \item $\beta_g \notin [0.5, 1.5]$ for any subgroup, or $\sign(\beta_g) \neq \sign(\beta_{\text{overall}})$ for any subgroup: \textbf{Severe concept shift. Walk Away} for that subgroup.
\end{itemize}

\paragraph{Connection to Jin et al.} \citet{jin2025beyondreweighting} demonstrate that covariate shift magnitude predicts concept shift severity: $|\hat{h}_{Y|X}| / |\hat{h}_X| \leq 1$ with high probability across 680 studies. This means Diagnostic~1 serves as an \emph{early warning} for Diagnostic~2: if covariate overlap is poor, conditional calibration failure is likely. We test this relationship empirically in \Cref{sec:empirical_anes}.

\begingroup
For instance, a retailer may want to target a promotion to older customers in rural markets even though only a small share of its 10,000 synthetic profiles matches that segment. The nominal panel size then overstates the information available for the decision, and weak overlap tells the research team to budget real interviews for that segment before using the synthetic estimate. This warning does not prove that the LLM's responses are wrong; it tells the team where conditional calibration deserves more scrutiny.
\endgroup

\begin{remark}[Practical subgroup definition]
Subgroups $\mathcal{G}$ should correspond to the demographic partitions relevant for the firm's decision. A CPG company segmenting by age and income defines $\mathcal{G}$ accordingly. The diagnostic is conservative: it flags failure if \emph{any} decision-relevant subgroup fails calibration.
\end{remark}

\subsection{Diagnostic 3: Cross-LLM Stability}
\label{sec:diag_stability}

\paragraph{Question.} Do different LLMs produce substantively similar conclusions?

\paragraph{Estimation.} Generate synthetic panels from $K \geq 2$ different LLMs (e.g., GPT-4o, Claude, Gemini) using identical prompts and demographic profiles. Compute pairwise rank correlations of subgroup-level estimates:
\begin{equation}
  \rho_{jk} = \SpearmanCorr\bigl(\hat{\tau}_{j,g}, \hat{\tau}_{k,g}\bigr) \quad \text{for LLMs } j, k \text{ across subgroups } g \in \mathcal{G}.
  \label{eq:cross_llm}
\end{equation}

\paragraph{Decision thresholds.}
\begin{itemize}[leftmargin=2em]
  \item $\min_{j,k} \rho_{jk} > 0.80$: \textbf{ Models agree on subgroup rankings. Proceed to conditional calibration before interpreting estimate levels.}
  \item $\min_{j,k} \rho_{jk} \in [0.50, 0.80]$: \textbf{Partial agreement.} Flag concept shift risk; proceed to conditional calibration.
  \item $\min_{j,k} \rho_{jk} < 0.50$: \textbf{Models disagree fundamentally. Walk Away.} Responses are model-specific artifacts, not population signals.
\end{itemize}

\paragraph{Rationale.} If different LLMs---trained on different data, with different RLHF procedures---produce the same subgroup patterns, the patterns likely reflect genuine population structure rather than model-specific artifacts. Disagreement signals concept shift.

\begingroup
The distinction between agreement in ranks and agreement in levels matters for a targeting decision. Three models may agree that purchase intent is highest among young urban customers and lowest among older rural customers while producing materially different estimates for every segment. The shared ranking shows that the ordering survives a change of model, but it does not establish the purchase-intent levels the firm needs for targeting; that still requires conditional calibration against real customers. If the models reverse which segment leads, averaging their estimates would hide the disagreement rather than recover a population signal.
\endgroup

\begin{remark}[Cost]
Cross-LLM stability requires only API calls to additional LLMs (no real data), making it the cheapest diagnostic after covariate overlap. It can be run before collecting any real-sample data.
\end{remark}

\begin{table}[htbp]
\centering
\caption{Diagnostic decision thresholds (operational defaults). Each diagnostic returns an interpretable metric that maps to one of three actions. Threshold sensitivity is examined in Online Appendix~C.5.}
\label{tab:thresholds}
\begin{tabular}{l l c c c}
\toprule
\textbf{Diagnostic} & \textbf{Metric} & \textbf{Proceed} & \textbf{Caution / Correct} & \textbf{Walk Away} \\
\midrule
Covariate overlap & $\ESS/N$ & $> 0.10$ & $[0.01, 0.10]$ & $< 0.01$ \\[3pt]
Conditional calibration & $\beta_g$ & $[0.5, 1.5]$, same sign & Outside range & Sign flip \\[3pt]
Cross-LLM stability & $\min(\rho_{jk})$ & $> 0.80$ & $[0.50, 0.80]$ & $< 0.50$ \\
\bottomrule
\end{tabular}
\end{table}
\FloatBarrier

\section{Doubly Robust Correction Module}
\label{sec:correction}

\begingroup
The diagnostics in \Cref{sec:diagnostics} can show that correction is inappropriate, but they cannot by themselves identify a corrected population mean. The Correct branch is available only when the firm has a complete invitation frame: it knows who was invited, who responded, and the frame variables used to model response for respondents and nonrespondents. Without that structure, the framework returns correction unavailable rather than silently treating a respondent-only sample as representative.
\endgroup

\subsection{Setup}
\label{sec:correction_setup}

\begingroup
Let $U=\{1,\ldots,M\}$ denote the firm's invitation frame. For every unit $i$, the firm observes covariates $X_i$, a fixed LLM prediction $f_i$, a nonnegative base or design weight $a_i$, and the response indicator $R_i$. The human outcome $Y_i$ is observed only when $R_i=1$. Here $Y_i$ is the value unit $i$ would report under a survey instrument, question order, and mode fixed before fielding. Let $W_i$ collect $X_i$, $f_i$, $a_i$, and any design variables used in the response and outcome models.

The primary estimand is the weighted mean over the realized frame:
\begin{equation}
  \theta_U = \frac{\sum_{i\in U} a_iY_i}{\sum_{i\in U} a_i},
  \label{eq:finite_frame_estimand}
\end{equation}
with $a_i=1$ when the frame is a census.

\paragraph{Human-data floor.}
Identification requires missing at random, $Y_i \indep R_i \mid W_i$, and positivity, $\pi(W_i)=\Pr(R_i=1\mid W_i)>0$, throughout the frame. Without MAR or another credible restriction linking respondent and nonrespondent outcomes, $\theta_U$ is not point identified: bounded outcomes admit only partial-identification bounds, and unrestricted unbounded outcomes need not admit informative bounds. The complete frame makes the response regression learnable and estimated positivity auditable; it does not make MAR or true positivity hold. If the frame is absent, MAR is not credible, or positivity fails, the decision is to collect better human data or Walk Away.
\endgroup

\subsection{Why the Previous Correction Failed}
\label{sec:aipw}

\begingroup
The previous equation was not a valid concept-shift correction. Substituting its own definition of the target outcome model gives
\begin{equation}
\begin{aligned}
Y_i^{\cS}+\hat g(X_i)-\hat\mu_{\cT}(X_i)
&=Y_i^{\cS}+\hat g(X_i)
  -\{\hat\mu_{\cS}(X_i)+\hat g(X_i)\}\\
&=Y_i^{\cS}-\hat\mu_{\cS}(X_i).
\end{aligned}
\label{eq:legacy_cancellation}
\end{equation}
The cancellation is exact for every observation and every value of $\hat g$. The estimated concept-shift function therefore cannot affect the augmentation or the final estimate. In the historical implementation, the weights were also normalized inside the augmentation, leaving a weighted source-residual term that behaved mainly as additional fold noise while the real-sample mean supplied the estimate's leading term.

This explains why the previous empirical correction often appeared successful: it anchored the result to observed human outcomes when the naive synthetic bias was large. It did not establish the claimed concept-shift correction or double robustness. We retain the legacy implementation only for historical replication and as a comparator in the simulation battery; its correction columns do not supply evidence for the frame-AIPW method introduced next.
\endgroup

\subsection{Frame-AIPW Estimator}
\label{sec:frame_aipw}

\begingroup
\begin{definition}[Frame-AIPW estimator]
\label{def:frame_aipw}
Let $A_M=\sum_{i\in U}a_i$, and let $k(i)$ denote the held-out fold containing unit $i$. The cross-fitted estimator is
\begin{equation}
\widehat{\theta}_U
=
\frac{1}{A_M}
\sum_{i\in U}a_i
\left[
\widehat{\mu}_{-k(i)}(W_i)
+
\frac{R_i}{\widehat{\pi}_{-k(i)}(W_i)}
\left\{Y_i-\widehat{\mu}_{-k(i)}(W_i)\right\}
\right].
\label{eq:frame_aipw}
\end{equation}
\end{definition}

The response model estimates $\pi(W_i)=\Pr(R_i=1\mid W_i)$ and is trained on every row in each training fold, using $R_i$ as the label. The outcome model estimates $\mu(W_i)=\E[Y_i\mid W_i]$ and is trained only on respondents in that training fold. Neither nuisance model sees the held-out fold.

The LLM prediction $f_i$ enters $W_i$ as a fixed auxiliary variable available for respondents and nonrespondents. Synthetic outcomes are not pooled with human outcomes in the primary outcome regression: doing so would require a separate transportability assumption. Base weights $a_i$ remain separate from inverse response probabilities.
\endgroup

\subsection{Theoretical Properties}

\begingroup
\begin{theorem}[Point robustness under two nuisance regimes]
\label{thm:double_robustness}
Let
\[
m(q,h)=h(W)+\frac{R}{q(W)}\{Y-h(W)\}.
\]
Conditional on the nuisance-training sample $\mathcal T$,
\[
\E[m(q,h)-Y\mid W,\mathcal T]
=
\left\{1-\frac{\pi(W)}{q(W)}\right\}
\{h(W)-\mu(W)\}.
\]
The two arms require different conditions; neither arm requires both nuisance models to be consistent.
\begin{enumerate}[label=(\alph*)]
  \item \textbf{Response-design mode.} If $q=\pi$ exactly, the estimator is response-design unbiased for $\theta_U$ for any fold-independent $h$, regardless of $h$'s consistency. With estimated propensities: cross-fitted $\widehat\pi\to\pi$, fold-independent bounded outcome predictions, and the weight regularity conditions (A4b) yield consistency for $\theta_U$. No outcome-model consistency is required.
  \item \textbf{Model-assisted mode.} Under the declared superpopulation model (frame units iid or iid clusters), if $h=\mu$ while $q$ may be misspecified but remains bounded away from zero, the estimator is consistent for $\theta_w=\E[aY]/\E[a]$. No propensity-model consistency is required. This is not an exact finite-frame unbiasedness claim for the realized frame.
\end{enumerate}
\medskip

\noindent The root-$M$ expansion, influence-function form, and normal-interval claims additionally require both nuisance estimators to satisfy the dual-consistency product-rate condition of Appendix~A; these inference conditions are stronger than either one-model point-consistency regime.
\end{theorem}

These are point-estimation statements. One correct nuisance model does not by itself guarantee root-$M$ inference or nominal confidence-interval coverage. The final battery therefore reports point bias and interval coverage separately; several outcome-correct/propensity-wrong cells under-cover.
\endgroup

\begingroup
\begin{proposition}[Uncertainty under two readings]
\label{prop:variance_readings}
Let
\[
\widehat m_i
=
\widehat\mu_i
+
\frac{R_i}{\widehat\pi_i}
\{Y_i-\widehat\mu_i\}.
\]
For the finite-frame response-design reading,
\begin{equation}
\widehat V_R
=
\frac{1}{A_M^2}
\sum_{i\in U}
a_i^2R_i
\frac{1-\widehat\pi_i}{\widehat\pi_i^2}
\{Y_i-\widehat\mu_i\}^2.
\label{eq:finite_frame_variance}
\end{equation}
For the superpopulation reading,
\begin{equation}
\widehat V_{\mathrm{SP}}
=
\frac{M}{M-1}
\frac{1}{A_M^2}
\sum_{i\in U}
a_i^2
\{\widehat m_i-\widehat\theta_U\}^2.
\label{eq:superpopulation_variance}
\end{equation}
\end{proposition}

The finite-frame formula conditions on the realized frame and measures uncertainty from which frame members respond. If everyone responds, $\widehat\theta_U=\theta_U$ exactly and $\widehat V_R=0$. The superpopulation sandwich instead treats frame units as draws and includes frame-to-frame variation. These formulas answer different questions and cannot be silently mixed.

Only under the declared superpopulation model, with both nuisance functions consistent and the product-rate condition satisfied, does the standard MAR efficient-influence-function result apply. That is not a claim of finite-frame efficiency or one-model-correct interval validity.
\endgroup

\subsection{Cross-Fitting}
\label{sec:crossfitting}

To prevent the ``double-dipping'' bias documented by \citet{song2026demystifying}---who show that using the same data for nuisance estimation and inference collapses coverage from 90\% to 47\%---we use $K$-fold cross-fitting.

\begingroup
A respondent's observed answer cannot help train the model that scores that same observation. Cross-fitting controls this own-row leakage; it does not, by itself, establish nuisance consistency, product-rate convergence, or valid inference.
\endgroup

\begingroup
\begin{algorithm}[htbp]

\caption{$K$-fold frame-AIPW estimation}
\label{alg:crossfitting}
\begin{algorithmic}[1]
\Require Frame $\{(W_i,R_i,R_iY_i,a_i)\}_{i=1}^M$; folds $K$; propensity floor $c$
\Ensure $\widehat\theta_U$; target-labeled standard error; positivity diagnostics
\State Partition frame indices $\{1,\ldots,M\}$ into $\mathcal I_1,\ldots,\mathcal I_K$
\For{$k=1,\ldots,K$}
  \State $\mathcal I_{-k}\gets\{1,\ldots,M\}\setminus\mathcal I_k$
  \State Fit $\widehat\pi_{-k}(W)$ on all $i\in\mathcal I_{-k}$ using $R_i$ as the label
  \State Fit $\widehat\mu_{-k}(W)$ only on $i\in\mathcal I_{-k}$ with $R_i=1$
  \State Predict $\widehat\pi_{-k}(W_i)$ and $\widehat\mu_{-k}(W_i)$ for $i\in\mathcal I_k$
  \If{$\min_{i\in\mathcal I_k}\widehat\pi_{-k}(W_i)<c$}
    \State Record positivity failure; withhold headline correction
  \EndIf
  \State Compute $\widehat m_i$ for each $i\in\mathcal I_k$
\EndFor
\State $\widehat\theta_U\gets A_M^{-1}\sum_{i\in U}a_i\widehat m_i$
\State Compute $\widehat V_R$ or $\widehat V_{\mathrm{SP}}$ according to the declared target reading
\State \Return estimate, target-labeled SE, and diagnostics
\end{algorithmic}
\end{algorithm}
\endgroup

\begingroup
\paragraph{Data availability determines the action.}
When the firm observes $X_i$, $f_i$, $a_i$, and $R_i$ for every invitation-frame unit and $Y_i$ for respondents, it may evaluate frame-AIPW after assessing whether MAR and positivity are credible. With respondent outcomes but no invitation frame, the firm cannot learn response propensities because it does not observe who was invited but did not respond. Frame correction is therefore unavailable. A survey mean may still be used if its own sampling design makes it representative; otherwise the firm must reconstruct the frame or collect better data. Aggregate benchmarks can support marginal diagnostics or sensitivity analysis, but not frame-AIPW correction. The respondent count determines how flexibly the nuisance models can be fit; it does not remove the frame, MAR, or positivity requirements. We therefore make no universal claim that a fixed number of human responses is sufficient.
\endgroup

\begingroup
\paragraph{Scope.}
The implemented and evaluated target in this paper is a scalar finite-frame mean. A prespecified segment mean can be estimated by redefining the frame and weights for that segment, but the segment still requires its own MAR and positivity assessment. Product rankings, segment treatment effects, and WTP or demand curves require additional estimands, multiplicity procedures, and shape assumptions; we treat those extensions as future work rather than claiming them as properties of the current implementation.
\endgroup

\section{Decision Framework: Trust, Correct, or Walk Away}
\label{sec:decision}

We synthesize the diagnostics from \Cref{sec:diagnostics} and the correction from \Cref{sec:correction} into a practitioner-oriented decision framework. The pipeline is ordered from cheapest to most expensive---practitioners can exit early.

\subsection{The Decision Flowchart}

\begingroup
\paragraph{What each verdict requires.} The three verdicts are not symmetric. Walk Away is the only verdict the cheap diagnostics can produce on their own; Trust and Correct both rest on human evidence, and Correct additionally rests on frame structure:

\begin{center}
\small
\begin{tabularx}{\textwidth}{>{\raggedright\arraybackslash}p{1.9cm} X}
\toprule
\textbf{Verdict} & \textbf{Requirements and scope} \\
\midrule
Walk Away & Can be triggered at any stage: poor overlap, fundamental cross-LLM disagreement, calibration sign flips, an absent invitation frame, not-credible MAR, or a failed positivity gate. No further data collection is implied beyond the stated remedy. \\[4pt]
Trust & Requires human calibration evidence showing subgroup slopes within $[0.5, 1.5]$ with no sign flips, plus adequate overlap. Diagnostics assess held-out respondent prediction; they cannot certify MAR or future model behavior. Trust means ``use the synthetic panel for this question family, this model version, now.'' \\[4pt]
Correct & Requires everything in Trust's evidence base plus the full correction data structure: covariates $X_i$, predictions $f_i$, weights $a_i$, and response indicators $R_i$ for all invited units, outcomes $Y_i$ for respondents, credible MAR given $W$, positivity with fitted propensities above the prespecified floor, and a declared target reading. If any element is missing, the framework returns \emph{correction unavailable / collect frame information} rather than routing to another estimator. \\
\bottomrule
\end{tabularx}
\end{center}

Zero-human diagnostics can end the process (Walk Away) but can never start it toward Trust: some human ground truth is always required before synthetic estimates are used.
\endgroup

\begingroup
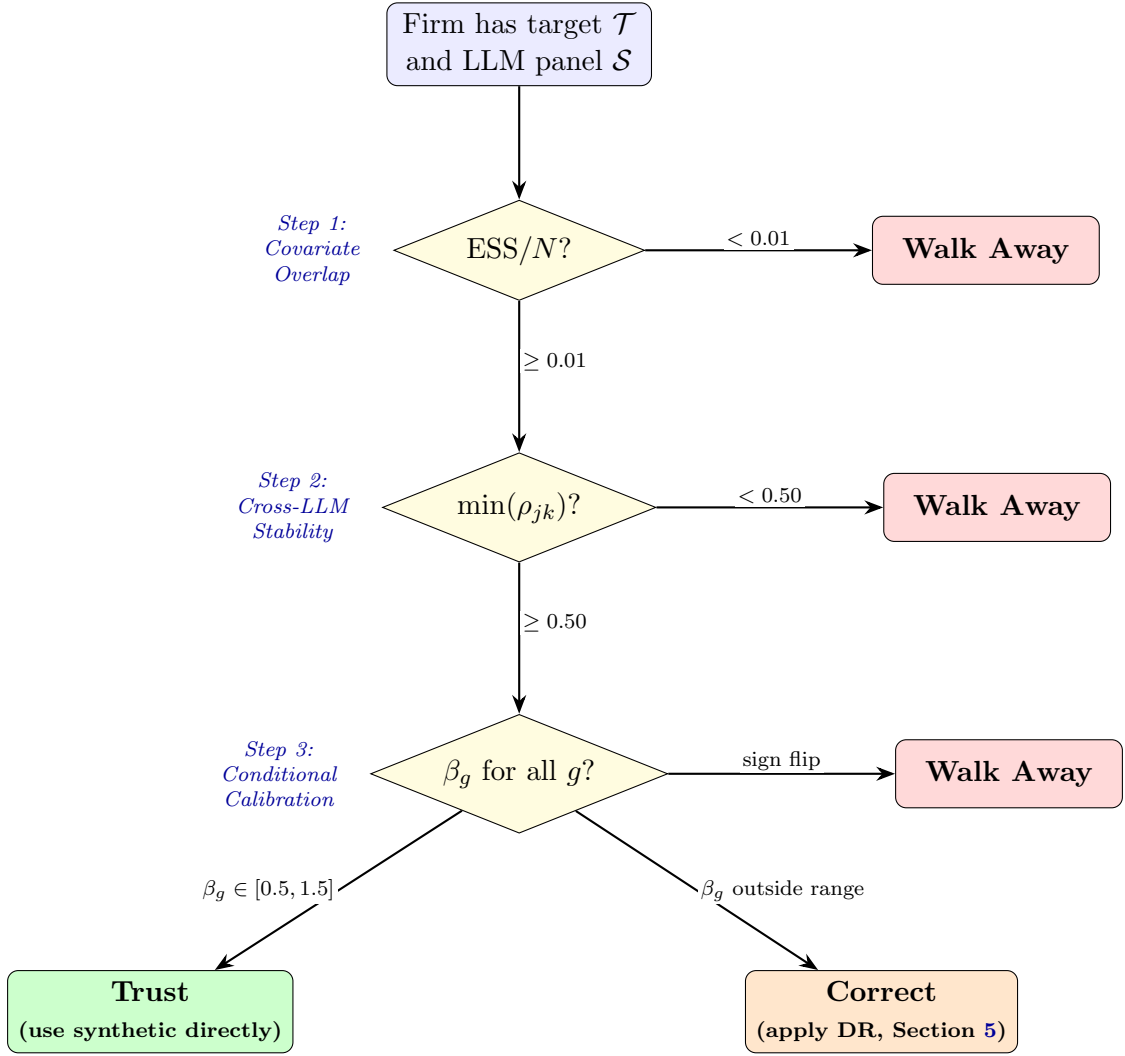
\begin{figure}[htbp]

\centering
\begin{tikzpicture}[
  node distance=1.2cm and 1.5cm,
  startstop/.style={rectangle, rounded corners, minimum width=3.5cm, minimum height=0.9cm, text centered, text=black, draw=black, fill=blue!8, font=\small, align=center},
  process/.style={rectangle, minimum width=4cm, minimum height=0.9cm, text centered, text=black, draw=black, fill=white, font=\small, align=center},
  decision/.style={diamond, minimum width=3cm, minimum height=1cm, text centered, text=black, draw=black, fill=yellow!15, aspect=2.5, align=center},
  outcome/.style={rectangle, rounded corners, minimum width=3cm, minimum height=0.9cm, text centered, text=black, draw=black, font=\small\bfseries, align=center},
  trust/.style={outcome, fill=green!20},
  correct/.style={outcome, fill=orange!20},
  walkaway/.style={outcome, fill=red!15},
  arrow/.style={-{Stealth[length=2.5mm]}, thick},
  label/.style={font=\scriptsize, text=black, fill=white, inner sep=1pt}
]

\node (start) [startstop] {Firm has target $\cT$\\and LLM panel $\cS$};
\node (d1) [decision, below=1.5cm of start] {$\ESS/N$?};
\node (d2) [decision, below=2cm of d1] {$\min(\rho_{jk})$?};
\node (d3) [decision, below=2cm of d2] {$\beta_g$ for all $g$?};

\node (wa1) [walkaway, right=3cm of d1] {Walk Away};
\node (wa2) [walkaway, right=3cm of d2] {Walk Away};
\node (wa3) [walkaway, right=3cm of d3] {Walk Away};

\node (trust) [trust, below left=2.2cm and 2cm of d3] {Trust\\{\scriptsize(use synthetic directly)}};
\node (d4) [decision, below right=2.2cm and 2cm of d3] {Invitation\\frame?};

\draw [arrow] (start) -- (d1);

\draw [arrow] (d1) -- node[label, above] {$< 0.01$} (wa1);
\draw [arrow] (d1) -- node[label, right, pos=0.4] {$\geq 0.01$} (d2);

\draw [arrow] (d2) -- node[label, above] {$< 0.50$} (wa2);
\draw [arrow] (d2) -- node[label, right, pos=0.4] {$\geq 0.50$} (d3);

\draw [arrow] (d3) -- node[label, above] {sign flip} (wa3);
\draw [arrow] (d3) -- node[label, left, pos=0.5] {\scriptsize$\beta_g \in [0.5, 1.5]$} (trust);
\draw [arrow] (d3) -- node[label, right, pos=0.5] {\scriptsize$\beta_g$ outside range} (d4);

\node (correct) [correct, below left=1.6cm and 0.4cm of d4] {Correct\\{\scriptsize(frame-AIPW, \Cref{sec:correction})}};
\node (unavail) [walkaway, below right=1.6cm and 0.4cm of d4, minimum width=3.8cm, align=center] {\scriptsize Correction unavailable\\\scriptsize (collect frame information)};
\draw [arrow] (d4) -- node[label, left, pos=0.5, align=center, text width=2.2cm] {\scriptsize frame complete, MAR credible, positivity holds} (correct);
\draw [arrow] (d4) -- node[label, right, pos=0.5, align=center, text width=2.4cm] {\scriptsize no frame / MAR not credible / positivity fails} (unavail);

\node[left=0.3cm of d1, font=\scriptsize\itshape, text=black, align=center] {Step 1:\\Covariate\\Overlap};
\node[left=0.3cm of d2, font=\scriptsize\itshape, text=black, align=center] {Step 2:\\Cross-LLM\\Stability};
\node[left=0.3cm of d3, font=\scriptsize\itshape, text=black, align=center] {Step 3:\\Conditional\\Calibration};
\node[right=0.3cm of d4, font=\scriptsize\itshape, text=black, align=center] {Step 4:\\Frame\\Check};

\end{tikzpicture}
\caption{Decision framework: Trust, Correct, correction-unavailable, or Walk Away. The pipeline is ordered from cheapest (Step~1: demographics only) to most expensive (Step~4: invitation-frame data). Practitioners can exit early at Steps~1--2 without collecting real data; Step~3 requires a human calibration sample; Step~4 requires the full invitation-frame structure before frame-AIPW (\Cref{sec:correction}) may run.}
\label{fig:decision_flowchart}
\end{figure}
\endgroup

\begingroup
\paragraph{Design principle.} The pipeline is ordered by cost: covariate overlap requires only demographics (CRM data), cross-LLM stability requires only API calls, conditional calibration requires a small real sample, and Step~4 requires the firm's invitation-frame records. Practitioners exit early if Steps~1--2 signal Walk Away. Steps~3--4 are not interchangeable: calibration evidence earns Trust or identifies concept shift, but only the frame check can make correction available. A failed frame check returns ``correction unavailable'' with a reason---it never silently substitutes another method (\Cref{sec:correction}).
\endgroup

\subsection{Worked Example: D2C Product Launch}
\label{sec:worked_example}

Consider a direct-to-consumer company testing purchase intent for a new biodegradable cleaning product across customer segments.

\emph{Step~1: Covariate Overlap.} The company's CRM shows its customer base skews female (65\%), urban (72\%), and higher-income (median \$85K). A domain classifier trained on CRM demographics vs.\ GPT-4o--generated profiles yields $\ESS/N = 0.15$ (adequate overlap). \emph{Proceed.}

\emph{Step~2: Cross-LLM Stability.} Synthetic panels from GPT-4o, Claude, and Gemini yield pairwise segment-level rank correlations of $\rho \in [0.79, 0.87]$. The minimum $\min(\rho_{jk}) = 0.79$ falls in $[0.50, 0.80]$---partial agreement, flagging concept shift risk. \emph{Proceed to Step~3.}

\emph{Step~3: Conditional Calibration.} A 200-person survey of actual customers yields calibration slopes by segment:

\begin{center}
\small
\begin{tabular}{l c l}
\toprule
\textbf{Segment} & $\beta_g$ & \textbf{Decision} \\
\midrule
Female, 25--34, urban & 1.12 & Trust \\
Male, 35--44, suburban & 0.73 & Correct \\
Female, 55+, rural & 0.48 & Walk Away \\
\bottomrule
\end{tabular}
\end{center}

\begingroup
The 55+ rural segment has $\beta_g = 0.48 < 0.5$, falling below the Trust/Correct boundary: for this segment the framework returns Walk Away, and the firm relies on its 200-person survey alone. The other two segments pass conditional calibration, but calibration evidence alone does not make correction available. Applying \Cref{sec:correction} requires the invitation frame: covariates and LLM predictions for every invited customer, the response indicator for each, respondent outcomes, credible MAR given those variables, and fitted propensities above the prespecified floor.

Suppose the firm's CRM records all 10,000 invited customers with demographics and stored GPT-4o predictions, and the survey was administered as an invitation to a random subsample of that frame. Then $R_i$ is observed for every invited unit, MAR is plausible given the recorded demographics, and no fitted propensity falls below the floor: for the two calibrated segments, the firm applies frame-AIPW and obtains corrected purchase-intent estimates with target-labeled standard errors.

Suppose instead the firm bought a stand-alone 200-person panel matched to its customer base, with no record of which customers were invited or declined. The diagnostics are unchanged, but correction is unavailable---there is no way to learn who self-selected out of the human sample. The framework reports this explicitly rather than substituting another estimator.
This example illustrates a key feature of the framework: different segments of the \emph{same} synthetic panel can warrant different decisions, and identical diagnostic results can lead to different actions depending on whether the correction data structure exists.
\endgroup

Online Appendix~F provides additional implementation detail.

\section{Empirical Applications}
\label{sec:empirical}

We validate the framework across three settings spanning outcome types (ordinal, discrete choice), domains (political attitudes, consumer pricing and behavioral economics), and data provenance (pre-existing LLM failures, purpose-built digital twins).  The ANES and Twin-2K data contain paired human and LLM responses but no real invitation frame or observed nonresponse process. We therefore use them in two distinct ways: as descriptive diagnostic applications and, for correction validation, as finite frames in which human outcomes are masked under prespecified MAR and MNAR response mechanisms. \Cref{tab:datasets} summarizes the datasets.

\begin{table}[htbp]
\centering
\caption{Summary of empirical datasets.}
\label{tab:datasets}
\small
\begin{tabularx}{\textwidth}{>{\raggedright\arraybackslash}p{2.8cm} >{\centering\arraybackslash}p{1.2cm} >{\raggedright\arraybackslash}p{2cm} >{\raggedright\arraybackslash}p{2.8cm} X}
\toprule
\textbf{Dataset} & \textbf{$N_{\text{real}}$} & \textbf{Outcome Type} & \textbf{LLM(s)} & \textbf{Role} \\
\midrule
ANES (Bisbee/Argyle) & 2,286\textsuperscript{a} & Ordinal (0--100) & GPT-3.5, GPT-4, Falcon-40B & Proof of concept: detect known failures \\[6pt]
Twin-2K-500 (Toubia et al.) & 2,058\textsuperscript{b} & Binary (buy/not buy) & GPT-4.1-mini & Consumer domain: pricing \& behavioral economics \\
\bottomrule
\end{tabularx}

\vspace{2pt}
\footnotesize{\textsuperscript{a}2016 wave with complete feeling thermometers; the full Argyle dataset contains 7,530 respondents across waves.   \textsuperscript{b}Our extraction contains 172,884 individually paired human--LLM observations across 500+ questions and 17 tasks (16 experiments and one pricing task).}
\end{table}

For each dataset, we compare five estimators: (1)~naive synthetic mean, (2)~IPW-only reweighting, (3)~calibration-only, (4)~DR/AIPW combining both (\Cref{sec:correction}), and (5)~real-data-only baseline (gold standard).

\begingroup
\subsection{Simulation Battery}
\label{sec:empirical_battery}

We evaluate frame-AIPW in a frozen simulation battery that separates five data regimes---probability sampling, MAR missingness, covariate shift (wrong WHO), concept shift (wrong WHAT), and outcome-dependent self-selection (MNAR)---crossed with four nuisance arms: both nuisance models correct, propensity only correct, outcome model only correct, and both wrong. The identified DGPs use 2,000 replications per cell; the MNAR regime uses 1,000 per cell as a documented expected-failure case. Every replication passes the identical realized dataset to seven prespecified estimators; \Cref{tab:battery_frame_aipw} reports frame-AIPW bias and coverage by cell, with the full seven-estimator comparison in Online Appendix~C.

Truth is the full-frame latent mean. Coverage uses Wilson 95\% Monte Carlo intervals; a cell whose interval lies wholly below 0.95 is labeled a failure. Two fitted-propensity floor events occurred among the probability-sampling replications and are disclosed as diagnostics.

\begin{table}[htbp]

\centering
\caption{Frame-AIPW simulation battery: point bias and interval coverage by DGP class and nuisance arm. Bias is in each DGP's outcome scale. Coverage failures are reported as failures; ``compatible'' means the Wilson 95\% Monte Carlo interval includes 0.95.}
\label{tab:battery_frame_aipw}
\scriptsize
\setlength{\tabcolsep}{4pt}
\begin{tabular}{l l l r r r l l}
\toprule
\textbf{DGP} & \textbf{$\pi$ arm} & \textbf{$\mu$ arm} & \textbf{Reps} & \textbf{Bias} & \textbf{MCSE} & \textbf{Coverage [95\% MC]} & \textbf{Status} \\
\midrule
Prob.\ sampling & correct & correct & 2{,}000 & 0.0009 & 0.0011 & 95.15\% [94.12, 96.01] & compatible \\
Prob.\ sampling & correct & wrong & 2{,}000 & $-0.0144$ & 0.0022 & 99.35\% [98.89, 99.62] & conservative \\
Prob.\ sampling & wrong & correct & 2{,}000 & 0.0004 & 0.0009 & 88.15\% [86.66, 89.49] & \textbf{failure} \\
Prob.\ sampling & wrong & wrong & 2{,}000 & 0.2706 & 0.0013 & 0.35\% [0.17, 0.72] & \textbf{failure} \\
MAR missing & correct & correct & 2{,}000 & 0.0003 & 0.0008 & 95.10\% [94.06, 95.96] & compatible \\
MAR missing & correct & wrong & 2{,}000 & $-0.0026$ & 0.0010 & 99.85\% [99.56, 99.95] & conservative \\
MAR missing & wrong & correct & 2{,}000 & 0.0003 & 0.0008 & 92.55\% [91.32, 93.62] & \textbf{failure} \\
MAR missing & wrong & wrong & 2{,}000 & 0.0599 & 0.0013 & 81.30\% [79.53, 82.95] & \textbf{failure} \\
Covariate shift & correct & correct & 2{,}000 & 0.0002 & 0.0011 & 95.65\% [94.67, 96.46] & compatible \\
Covariate shift & correct & wrong & 2{,}000 & $-0.0297$ & 0.0022 & 99.50\% [99.08, 99.73] & conservative \\
Covariate shift & wrong & correct & 2{,}000 & 0.0001 & 0.0008 & 94.15\% [93.03, 95.10] & compatible \\
Covariate shift & wrong & wrong & 2{,}000 & 0.2647 & 0.0013 & 1.05\% [0.69, 1.60] & \textbf{failure} \\
Concept shift & correct & correct & 2{,}000 & $-0.0000$ & 0.0006 & 95.05\% [94.01, 95.92] & compatible \\
Concept shift & correct & wrong & 2{,}000 & $-0.0080$ & 0.0007 & 99.90\% [99.64, 99.97] & conservative \\
Concept shift & wrong & correct & 2{,}000 & 0.0000 & 0.0006 & 93.45\% [92.28, 94.45] & \textbf{failure} \\
Concept shift & wrong & wrong & 2{,}000 & 0.2171 & 0.0010 & 0.10\% [0.03, 0.36] & \textbf{failure} \\
Self-selection (MNAR) & correct & correct & 1{,}000 & 0.2338 & 0.0009 & 0.00\% [0.00, 0.38] & \textbf{failure} \\
Self-selection (MNAR) & correct & wrong & 1{,}000 & 0.2338 & 0.0009 & 0.00\% [0.00, 0.38] & \textbf{failure} \\
Self-selection (MNAR) & wrong & correct & 1{,}000 & 0.2337 & 0.0009 & 0.00\% [0.00, 0.38] & \textbf{failure} \\
Self-selection (MNAR) & wrong & wrong & 1{,}000 & 0.2328 & 0.0015 & 0.20\% [0.05, 0.73] & \textbf{failure} \\
\bottomrule
\end{tabular}
\end{table}

Three patterns matter for practice. First, point bias is negligible whenever either nuisance model is correct in the identified DGPs---the two one-model point-consistency modes of \Cref{thm:double_robustness} behave as stated. Second, point consistency does not deliver nominal intervals: every outcome-correct/propensity-wrong cell under-covers (88--94\%), which is why \Cref{sec:correction} separates point claims from inference claims. Third, under outcome-dependent self-selection frame-AIPW carries bias of roughly 0.23 with near-zero coverage in \emph{all} nuisance arms: when selection depends on the missing outcome itself, no nuisance model rescues identification, and the framework's answer is Walk Away.
\endgroup

\subsection{Semi-Synthetic Nonresponse Validation}
\label{sec:semi_synthetic}

\begingroup
\paragraph{Design and reporting rule.}
We treat each complete paired dataset as a finite frame, define truth as the full-frame human mean, and mask human outcomes under mild MAR, strong MAR, and outcome-dependent MNAR response mechanisms. This semi-synthetic nonresponse experiment retains the observed covariates, human outcomes, and LLM prediction errors while controlling only which human outcomes are observed. A frame-based AIPW (frame-AIPW) estimate is reportable only when every cross-fitted estimated response probability is at least 0.001. Otherwise, we mark the estimate unavailable rather than silently stabilize extreme inverse-probability weights. This gate operationalizes Walk Away: it indicates inadequate estimated support for stable correction, not a numerical error.

\paragraph{Unavailable estimates.}
Across 12,000 frame-AIPW attempts, the positivity gate withheld 99 estimates (0.83\%): 14 in ANES and 85 in Twin-2K. Every unavailable estimate arose because an estimated response probability fell below the prespecified threshold. No other estimator produced a non-finite numerical result, although numerical availability does not establish identification. Bias, RMSE, mean standard error, and interval width are calculated among reportable estimates. Coverage uses all attempted replications, with unavailable estimates counted as failures to cover.

\begin{table}[htbp]

\centering
\caption{Frame-AIPW under controlled masking: bias and coverage by outcome and response regime (1{,}000 replications per cell). Bias is among reportable estimates, in each outcome's scale; coverage uses all attempts, counting unavailable estimates as failures to cover.}
\label{tab:semi_frame_aipw}
\scriptsize
\setlength{\tabcolsep}{4pt}
\begin{tabular}{l l l r r r l l}
\toprule
\textbf{Dataset} & \textbf{Outcome} & \textbf{Regime} & \textbf{Reportable} & \textbf{Bias} & \textbf{MCSE} & \textbf{Coverage [95\% MC]} & \textbf{Status} \\
\midrule
ANES & Dem.\ thermometer & Mild MAR & 1{,}000/1{,}000 & 0.0149 & 0.0130 & 95.80\% [94.37, 96.88] & compatible \\
ANES & Dem.\ thermometer & Strong MAR & 992/1{,}000 & 0.0041 & 0.0251 & 94.90\% [93.36, 96.10] & compatible \\
ANES & Dem.\ thermometer & MNAR & 1{,}000/1{,}000 & 5.2855 & 0.0166 & 0.00\% [0.00, 0.38] & \textbf{failure} \\
ANES & Rep.\ thermometer & Mild MAR & 1{,}000/1{,}000 & 0.0061 & 0.0130 & 95.70\% [94.26, 96.79] & compatible \\
ANES & Rep.\ thermometer & Strong MAR & 995/1{,}000 & $-0.0254$ & 0.0232 & 96.10\% [94.71, 97.13] & compatible \\
ANES & Rep.\ thermometer & MNAR & 999/1{,}000 & 5.7040 & 0.0302 & 0.40\% [0.16, 1.02] & \textbf{failure} \\
Twin-2K & Purchase rate & Mild MAR & 1{,}000/1{,}000 & 0.0001 & 0.0001 & 95.40\% [93.92, 96.53] & compatible \\
Twin-2K & Purchase rate & Strong MAR & 999/1{,}000 & 0.0004 & 0.0002 & 96.70\% [95.40, 97.64] & conservative \\
Twin-2K & Purchase rate & MNAR & 928/1{,}000 & 0.0180 & 0.0002 & 4.20\% [3.12, 5.63] & \textbf{failure} \\
Twin-2K & Behav.\ econ.\ score & Mild MAR & 1{,}000/1{,}000 & $-0.0002$ & 0.0002 & 94.70\% [93.13, 95.93] & compatible \\
Twin-2K & Behav.\ econ.\ score & Strong MAR & 999/1{,}000 & $-0.0006$ & 0.0003 & 95.20\% [93.69, 96.36] & compatible \\
Twin-2K & Behav.\ econ.\ score & MNAR & 989/1{,}000 & 0.0440 & 0.0004 & 1.40\% [0.84, 2.34] & \textbf{failure} \\
\bottomrule
\end{tabular}
\end{table}

\Cref{tab:semi_frame_aipw}  evaluates the estimator on real paired data with imposed response masking.  Under mild or strong imposed MAR masking, absolute bias is at most 0.0254 thermometer points for ANES, 0.04 percentage points for the Twin-2K purchase rate, and 0.0006 mean response-position units for the Twin-2K behavioral-economics score. Seven of the eight coverage estimates are compatible with 95\%; the Twin-2K purchase-rate estimate under strong MAR is conservative at 96.7\%. Under outcome-dependent MNAR the estimator fails materially: bias of 5.3--5.7 thermometer points on ANES and 1.8--4.4 percentage points on Twin-2K purchase behavior, with coverage collapsing to 0--4\%. This is the empirical Walk Away case: when response depends on the missing human outcome itself, no amount of frame information restores point identification.
\endgroup

\subsection{ANES: Diagnostic Illustration}
\label{sec:empirical_anes}

\begingroup
This application illustrates the diagnostics on known failures. It contains no invitation frame or response mechanism, so it provides no evidence about correction under real nonresponse; estimator validation is the battery's and the masking experiment's job (\Cref{sec:empirical_battery,sec:semi_synthetic}).
\endgroup

\paragraph{Data.} We use the replication data from \citet{bisbee2024synthetic}, which contains ANES 2016 respondents ($n = 2{,}286$ after restricting to 2016 with complete feeling thermometers) paired with synthetic thermometer scores from Falcon-40B and GPT-4. For cross-LLM stability, we compare three models: GPT-3.5, GPT-4, and Falcon-40B.\footnote{GPT-3.5 responses lack respondent identifiers for individual-level pairing but can be aggregated to party-subgroup means for cross-LLM stability analysis. We use Falcon-40B, which is merged at the respondent level, as the primary synthetic panel for the diagnostic framework analysis.}

\paragraph{Known failures.} \citet{bisbee2024synthetic}   document substantially smaller standard deviations for synthetic estimates than for estimates based on ANES responses, significant differences in 48\% of coefficients, sign flips in 32\% of those cases, and systematic subgroup failures. Our diagnostic framework should detect these failures ex ante.

\paragraph{Setting 1: Full-sample analysis (no covariate shift).}
Since the synthetic respondents use the same demographic profiles as real ANES respondents by construction, any estimation error is pure concept shift. We run the diagnostic framework with Falcon-40B as the synthetic source and ANES 2016 humans as the target, using feeling thermometers toward the Democratic Party and Republican Party.

\emph{Results:} As expected, covariate overlap is near-perfect (ESS/N $= 0.954$). The diagnostics detect severe heterogeneous concept shift. Conditional calibration reveals sign flips in the mean bias direction across party subgroups:

\begin{table}[htbp]
\centering
\caption{ANES conditional calibration: party subgroup bias in Falcon-40B thermometers.}
\label{tab:anes_calibration}
\small
\begin{tabular}{l c c c c}
\toprule
\textbf{Thermometer} & \textbf{Democrats} & \textbf{Independents} & \textbf{Republicans} & \textbf{Decision} \\
\midrule
Dem.\ Party & $-18.7$~pp & $-3.7$~pp & $+3.0$~pp & Walk away \\
Rep.\ Party & $+1.4$~pp & $-8.2$~pp & $-11.0$~pp & Walk away \\
\bottomrule
\end{tabular}
\end{table}

\noindent The calibration slopes are far from 1 ($\beta_g \in [-0.004, 0.274]$, including a sign reversal), and---critically---the \emph{direction} of synthetic bias flips across subgroups. Falcon-40B systematically underestimates Democratic Party thermometers for Democrats ($-18.7$~pp) but slightly overestimates them for Republicans ($+3.0$~pp). This directional heterogeneity is exactly the pattern that makes uniform correction unreliable and that motivates the ``walk away'' recommendation.

The bias decomposition confirms: concept shift accounts for 96.8\% of total error (Democratic Party) and 97.3\% (Republican Party), with covariate shift at $< 3\%$---consistent with the matched-demographics design.

\begin{table}[htbp]
\centering
\caption{ANES estimator comparison on subsampled targets (Falcon-40B synthetic panel).}
\label{tab:anes_estimators}
\small
\begin{tabular}{l l c c c c c c}
\toprule
\textbf{Target} & \textbf{Therm.} & \textbf{$\tau_{\text{real}}$} & \textbf{$|\text{Bias}_{\text{naive}}|$} & \textbf{$|\text{Bias}_{\text{IPW}}|$} & \textbf{$|\text{Bias}_{\text{DR}}|$} & \textbf{Red.} & \textbf{CI} \\
\midrule
Young Dem. & Dem.\ Party & 72.8 & 33.1 & 21.6 & \textbf{0.2} & 99.5\% & $\checkmark$ \\
Older Rep. & Rep.\ Party & 66.2 & 28.8 & 14.5 & \textbf{2.0} & 92.9\% & $\checkmark$ \\
\bottomrule
\end{tabular}
\end{table}

\begingroup
\noindent In both cases, the framework recommends ``correct'' at the diagnostic stage (ESS/N $= 0.11$ and $0.26$, no sign flips within the homogeneous subgroups). \emph{Correction itself cannot be exercised on these data}: there is no invitation frame recording which respondents were invited and which declined, so response propensities are not learnable (\Cref{sec:correction}). Under the historical two-sample estimator studied before this revision---whose correction term cancels algebraically (\Cref{sec:aipw})---the reported ``corrected'' estimates were dominated by the real-sample mean; we therefore do not interpret its retained bias-reduction columns as correction evidence. The decomposition split below is a descriptive diagnostic quantity computed without any correction step.
\endgroup

The decomposition reveals a clean split: for young Democrats, covariate shift accounts for 34.8\% of total bias and concept shift for 65.2\%; for older Republicans, the split is 49.5\%/50.5\%. This quantification directly informs the practitioner: even after reweighting the panel to match the firm's customer demographics, substantial concept shift remains that the calibration model must correct.

\paragraph{Cross-LLM stability.}
All three LLMs (GPT-3.5, GPT-4, Falcon-40B) agree on the ranking of party subgroup thermometers ($\rho = 1.0$ for all pairwise comparisons) despite differing substantially in magnitude (\Cref{tab:anes_cross_llm}). With only three subgroups, any monotonic ordering yields $\rho = 1.0$, limiting statistical power; finer-grained subgroups would provide a more informative test. The models exhibit different concept shift patterns but identical structural rankings---precisely the signal cross-LLM stability is designed to detect. Extended results appear in Online Appendix~D.

\begin{table}[htbp]
\centering
\caption{Cross-LLM subgroup estimates: feeling thermometer toward Democratic Party.}
\label{tab:anes_cross_llm}
\small
\begin{tabular}{l c c c c}
\toprule
\textbf{Subgroup} & \textbf{Human (ANES)} & \textbf{GPT-3.5} & \textbf{GPT-4} & \textbf{Falcon-40B} \\
\midrule
Democrat       & 73.4 & 90.3 & 84.2 & 55.1 \\
Independent    & 42.2 & 58.7 & 56.6 & 39.3 \\
Republican     & 21.7 & 19.0 & 33.0 & 24.7 \\
\bottomrule
\end{tabular}
\end{table}

\subsection{Twin-2K-500: Consumer Pricing Domain}
\label{sec:empirical_twin2k}

\begingroup
As with ANES, this application illustrates the diagnostics on real paired data. It contains no invitation frame or response mechanism, so its correction-related columns describe the historical two-sample estimator and are not evidence of correction under real nonresponse.
\endgroup

\begingroup
\paragraph{Data.} We use the Twin-2K-500 dataset \citep{toubia2025digitaltwin}, published in \emph{Marketing Science}. It pairs 2,058 U.S. participants with individualized GPT-4.1-mini digital twins across more than 500 questions and 17 tasks: 16 experiments and one pricing task. Our extraction contains 172,884 individually paired human--LLM observations. Unlike the ANES application, which tests the framework on \emph{known failures}, Twin-2K-500 evaluates a recent digital-twin system. Our code computes 83.9\% item-level agreement for the product-pricing subset.
\endgroup

Our primary analysis uses the ``Product Preferences --- Pricing'' block (82,320 paired observations across 40 grocery products), where participants indicate whether they would purchase a specific product at a given price (binary; we recode to buy = 1, not buy = 0). Demographics include age group (18--29, 30--49, 50--64, 65+), sex, education (6 levels), and income (6 levels). We also analyze behavioral economics blocks spanning the conjunction fallacy, anchoring, framing effects, risk preferences, and time preferences.

\paragraph{ Matched profiles.}  Digital twins use the same demographic profiles as their human counterparts, so there is no covariate shift by construction ($\ESS/N = 1.0$). Any estimation error is \emph{pure concept shift}. The framework must detect  whether our computed 83.9\% item-level agreement is sufficient for reliable inference---and what happens when the target population differs from the source demographics.

\paragraph{Full-sample product pricing results.} With individually-paired data and identical demographics, the framework correctly routes the full-sample analysis to \textbf{Trust}:

\begin{table}[htbp]
\centering
\small

\caption{Twin-2K-500 product pricing: diagnostic results. ``Red.'' columns report bias reduction from the legacy two-sample estimator analyzed in Section~5; these values do not provide evidence for frame-AIPW.}
\label{tab:twin2k_pricing}
\begin{tabular}{l c c c c c c}
\toprule
\textbf{Analysis} & \textbf{ESS/N} & \textbf{$\beta_g$ range} & \textbf{Decision} & \textbf{$|\text{Bias}_{\text{naive}}|$} & \textbf{Red.} & \textbf{CI} \\
\midrule
Full sample (by age)           & 1.000 & $[0.86, 0.90]$ & Trust   & 0.010 & 99.2\% & \checkmark \\
Full sample (by race)          & 1.000 & $[0.80, 0.93]$ & Trust   & 0.010 & 99.2\% & \checkmark \\
Full sample (by income)        & 1.000 & $[0.87, 0.90]$ & Trust   & 0.010 & 99.2\% & \checkmark \\
High-income target             & 0.401 & $[0.21, 0.28]$ & Correct & 0.019 & 83.3\% & \checkmark \\
Young consumers (18--29)       & 0.203 & $[0.16, 0.54]$ & Correct & 0.040 & 94.0\% & \checkmark \\
\bottomrule
\end{tabular}
\end{table}

\begingroup
\noindent The full-sample results illustrate the diagnostic framework on a positive case: calibration slopes of $\beta_g \in [0.80, 0.93]$ fall comfortably within the Trust range ($[0.5, 1.5]$). The naive synthetic mean underestimates the human purchase rate by 1.0 percentage point (43.1\% vs.\ 44.1\%). We report these diagnostics descriptively; the ``Reduction'' column belongs to the retracted legacy estimator and is not interpreted as correction evidence, and Trust here means the diagnostics found no flagged failure on this sample---not a certified guarantee (\Cref{sec:decision}).
\endgroup

\begingroup
\paragraph{The target population can change a Trust decision.} Subgroup calibration---not aggregate agreement alone---supports the full-sample Trust result. When the target differs from the panel demographics, the same dataset routes to Correct. Practitioners therefore need diagnostics for each use case rather than carry the full-sample result into subgroup targeting. To our knowledge, this is the first formal subgroup-level validation of an LLM consumer panel.
\endgroup

\begingroup
\paragraph{Subgroup targeting introduces covariate shift.} The critical finding emerges when the target population differs from the LLM panel demographics. A firm targeting high-income consumers ($\geq$\$75K) faces ESS/N $= 0.401$---only 40\% effective overlap---and calibration slopes collapse to $\beta_g \in [0.21, 0.28]$, far below the Trust range: the diagnostics route this case out of Trust. Targeting young consumers (18--29) shows even lower overlap (ESS/N $= 0.203$). Whether either targeting scenario can proceed to correction is a Step-4 question (\Cref{sec:decision}): it depends on whether the firm holds an invitation frame for the targeted segment, not on the diagnostic values alone.
\endgroup

\paragraph{Uniform demographic calibration.} Unlike ANES, where calibration slopes vary sharply by age, the Twin-2K-500 full-sample analysis shows \emph{uniform calibration}: all subgroup slopes fall within $[0.80, 0.93]$ across age groups (0.856--0.905), income levels (0.870--0.903), and gender (0.874--0.884), with every subgroup routing to Trust. This contrast suggests that concept shift severity is task- and model-specific: GPT-4.1-mini, trained with 500+ prior responses per person, achieves good calibration on purchase decisions but fails on cognitive-bias tasks (below).

\paragraph{Behavioral economics blocks.} The Twin-2K-500 dataset includes 19 behavioral economics experiments spanning cognitive biases, risk preferences, and social preferences. We apply the diagnostic framework to each block, using the Spearman correlation between human and LLM person-level responses as a complementary measure of concept shift:

\begin{table}[htbp]
\centering
\small
\caption{Twin-2K-500: behavioral economics diagnostics (selected blocks).}
\label{tab:twin2k_behav}
\begin{tabular}{l c c c c}
\toprule
\textbf{Block} & \textbf{$\rho$} & \textbf{$\min(\beta_g)$} & \textbf{Decision} & \textbf{Interpretation} \\
\midrule
Product Pricing          & 0.860 & $[0.80, 0.93]$ & Trust & Purchase patterns tracked \\
WTP/WTA (Thaler)         & 0.587--0.692 & 0.46--0.61 & Correct/Trust & Risk framing matters \\
Proportion Dominance     & 0.605--0.689 & 0.47--0.61 & Correct/Trust & Moderate concept shift \\
False Consensus          & 0.857 & 0.74 & Trust & Demographic patterns tracked \\
Conjunction Fallacy      & 0.387--0.392 & 0.25--0.28 & Correct & LLM hyper-rational \\
Anchoring Effects        & 0.337--0.500 & 0.28--0.37 & Correct & LLM anchor-resistant \\
\bottomrule
\end{tabular}
\end{table}

\noindent The behavioral economics results reveal a clear pattern: LLM digital twins succeed on tasks where responses follow systematic demographic patterns (product pricing, risk gambles) but fail on tasks requiring \emph{cognitive biases}. The conjunction fallacy (Linda problem) and anchoring effects---both well-documented human irrationalities---show weak human--LLM correlation because GPT-4.1-mini is ``hyper-rational'' \citep{peng2025funhouse}: it applies logical reasoning where humans rely on heuristics. This finding has direct practical implications: synthetic panels are least reliable precisely for the behavioral economics phenomena that are most interesting to consumer researchers.

\begingroup
\paragraph{Sample-size evidence.} The historical sample-size sweep evaluated the retracted two-sample estimator and is not evidence for frame-AIPW. The frozen battery and controlled-masking experiments hold their frame and response designs fixed; they do not identify a universal respondent count at which diagnostics, point correction, or confidence intervals become reliable. We therefore report no minimum $n$. Applications must assess held-out calibration, positivity, influence-score behavior, and uncertainty for the specific frame and estimand.
\endgroup

\subsection{Cross-Dataset Comparison}
\label{sec:empirical_cross}

\begingroup
The cross-dataset comparison shows how the source of shift changes the decision (full results in Online Appendix~K). With matched demographics, concept shift accounts for 96.8--97.3\% of total error in ANES and 100\% in Twin-2K-500. Its severity is task-specific: ANES shows acute age-based failures ($\beta_g \in [-0.004, 0.274]$) and routes to Walk Away or Correct, whereas Twin-2K-500 pricing calibrates well across demographic groups ($\beta_g \in [0.80, 0.93]$) and routes to Trust but performs poorly on cognitive-bias tasks ($r = 0.34$--$0.50$). When Twin-2K-500 targets specific demographics, covariate shift becomes substantial---56\% for high-income targeting---and the decision moves to Correct.
\endgroup

\section{Discussion}
\label{sec:discussion}

\begingroup
\paragraph{Why validation remains task-specific.}
Better models do not eliminate the need for validation. In Twin-2K, the same GPT-4.1-mini system applied to the same people reaches $\rho = 0.86$ on product pricing but about $\rho = 0.39$ on the conjunction fallacy. The difference is task-specific: reproducing demographic purchase patterns is not the same problem as reproducing a human reasoning error.

A favorable diagnostic result means that the tested sample did not reveal the prespecified overlap, cross-model, or calibration failure. It does not certify that the synthetic outcome is true, that MAR holds, or that the result will survive a model or prompt update. Validation therefore attaches to a question family, model version, prompt configuration, target population, and time. A consequential change to any of them requires a new diagnostic pass; correction also requires a new frame check, nuisance fit, and positivity assessment.
\endgroup

\begingroup
\subsection{Limitations}

\paragraph{1. The empirical applications do not contain a real observed nonresponse process.}
ANES and Twin-2K contain paired human and LLM outcomes, but neither records a genuine invitation frame, who was invited, or who declined. The masking experiments therefore test frame-AIPW against known full-frame truth under imposed MAR and MNAR mechanisms. They show how the estimator behaves when those mechanisms are controlled; they do not establish that either dataset experienced the simulated response process or that MAR holds in an operational survey. A real-world correction claim still requires the firm's own frame, response indicators, positivity evidence, and a defensible substantive argument for MAR.

\paragraph{2. The Jin et al.\ bounding relationship may not transfer to LLM populations.}
\citet{jin2025beyondreweighting}  study the relationship between covariate and conditional shift across 680 replication studies under a random distribution-shift model. LLM-to-human shift can be more structured because model training and alignment induce shared response patterns. If the relationship fails, the overlap diagnostic remains informative about wrong WHO but cannot serve as an early warning for wrong WHAT.

\paragraph{3. Decision thresholds are operational defaults.}
The calibration-slope range $\beta_g\in[0.5,1.5]$, the cross-model correlation thresholds, and the overlap thresholds are not universal constants or consequences of the frame-AIPW theory. The frozen simulation battery evaluates estimator behavior under specified response and nuisance regimes; it does not validate the decision thresholds. Applications should prespecify thresholds in the outcome's decision scale and examine whether the resulting action changes under nearby choices.

\paragraph{4. Model and prompt changes limit transport over time.}
Model updates, prompt changes, and decoding changes can alter synthetic outputs \citep{bisbee2024synthetic,morris2025limitssynthetic}. A validation result applies to the question family, model version, and prompt configuration that produced it. A consequential update requires new diagnostics and, when correction is used, new nuisance fits and positivity checks.

\paragraph{5. No universal minimum human sample size is established.}
The amount of human data needed depends on the response rate, frame size, weight concentration, outcome variability, nuisance complexity, and the subgroup or estimand being reported. Neither the frozen battery nor the controlled-masking experiments identify a fixed respondent count that guarantees reliable diagnostics, point correction, or confidence intervals. The relevant checks are empirical: held-out calibration, positivity, influence-score behavior, and uncertainty under the declared target reading.

\paragraph{6. Fine-tuning and cross-LLM interactions remain unanalyzed.}
We do not study how pre-correction fine-tuning changes the diagnostics or frame-AIPW nuisance models. Twin-2K also provides only one LLM (GPT-4.1-mini), so cross-LLM stability cannot be evaluated in that consumer application. Multi-model consumer frames are a separate empirical extension.

\paragraph{7. Shared generation contexts change the dependence structure.}
Frame-AIPW treats the stored prediction $f_i$ as a fixed auxiliary variable; it does not require positive-temperature or observation-wise independent generation. Shared prompts, conversational memory, or batch-level context can nevertheless induce dependence across units and create leakage if related units are split across folds. Such units must be assigned to folds as groups, the shared-context structure must be disclosed, and cluster-level inference is required when the target sampling or response design is clustered.
\endgroup

\begingroup
\subsection{Data and Decision Costs}

The framework is a data-escalation rule, not a claim that synthetic users eliminate human-data costs. Overlap and cross-model diagnostics can rule out a synthetic panel before a human study is fielded, but Trust still requires human calibration evidence. Correct requires more: an invitation frame, response indicators for all invited units, and respondent outcomes. When those records do not exist, the relevant cost is the cost of reconstructing the frame or collecting new human data, not the cost of running another estimator. Dollar costs depend on the population, instrument, response rate, and recruitment design; we do not claim universal savings.

\paragraph{What correction can and cannot buy.}
Frame-AIPW is not guaranteed to be more precise than a design-valid human-only estimator. Under the finite-frame reading, response-design variance depends on the response probabilities and on residual variation after conditioning on $W$; a predictive stored LLM outcome can help through the outcome model, but it does not guarantee a variance reduction. At full response, the estimator recovers the full-frame mean exactly and response-design variance is zero.

With incomplete response, the relevant question is whether the frame, MAR, and positivity make the target identifiable and whether the target-labeled interval is credible---not whether the retracted estimator beats a respondent mean. The final battery supports one-model point robustness in identified regimes, but it also records undercoverage when only the outcome model is correct and failure under MNAR. We therefore make no universal efficiency, RMSE-dominance, or precision-gain claim.

\subsection{Extensions}

Three extensions follow from the remaining identification gaps. First, partial-identification and sensitivity analyses could replace point correction when MAR is not credible. Second, repeated validation could monitor task-, model-, and prompt-specific drift over time. Third, product rankings, segment effects, willingness-to-pay curves, and adaptive recruitment require new estimands or decision criteria; they are future methods, not properties of the scalar frame-mean estimator evaluated here.
\endgroup

\section{Conclusion}
\label{sec:conclusion}

\begingroup
This paper turns synthetic-user validation into a diagnosis-and-action problem. The diagnostics distinguish wrong WHO from wrong WHAT and can rule out a synthetic panel before correction is attempted. The Correct branch is narrower: it requires an invitation frame, response indicators for all invited units, respondent outcomes, credible MAR, and adequate positivity. When those conditions hold, frame-AIPW targets the declared finite-frame or superpopulation mean; when they do not, the framework returns correction unavailable or Walk Away.

The final experiments show why the distinction matters. In the identified simulation regimes, frame-AIPW has small point bias when either the response model or outcome model is correct, but interval coverage is not generally one-model robust: the outcome-correct, propensity-wrong cells cover only 88--94\%. Under outcome-dependent MNAR, all nuisance arms fail, with bias around 0.23 and near-zero coverage. Controlled masking on ANES and Twin-2K yields negligible bias and coverage that is near nominal or conservative under imposed MAR, but failure under MNAR; 99 of 12,000 frame-AIPW attempts are withheld by the prespecified positivity gate. Because neither dataset contains a genuine observed nonresponse process, these are estimator-validation results, not claims that frame-AIPW corrected real nonresponse in either study.

Synthetic predictions are therefore auxiliary information, not a replacement for human identification. Trust is scoped to the tested question family, model version, prompt configuration, target population, and time. Correct additionally requires the response frame and its assumptions. Without that support, the defensible result is not a stabilized estimate---it is a decision to collect better human data or Walk Away.
\endgroup

\clearpage
\appendix
\section*{Online Appendix}
\numberwithin{equation}{section}
\numberwithin{figure}{section}
\numberwithin{table}{section}

\section{Proofs}
\label{app:proofs}

\begingroup
This appendix provides the frame-AIPW derivations used in Section~5 of the main text.

\subsection{Assumptions}

\begin{enumerate}[label=(\textbf{A\arabic*}),leftmargin=3em]
  \item \textbf{Well-defined outcome.} The survey instrument, question order, and mode are fixed before fielding, so $Y_i$ denotes the outcome frame unit $i$ would report under that instrument.
  \item \textbf{Missing at random.} $Y\indep R\mid W$.
  \item \textbf{Positivity.} $\pi(W)\geq\varepsilon>0$ on the frame, and cross-fitted propensities remain above a prespecified $c$ with $0<c<\varepsilon$.
  \item \textbf{Response design.} Conditional on the realized finite frame, response indicators are independent Bernoulli draws with probabilities $\pi_i$, unless a declared cluster design and cluster-level inference replace this assumption.
  \item \textbf{Generation hygiene.} The prediction $f$ is frozen and available for respondents and nonrespondents, with no target-$Y$ or response-label leakage. Any target-data adaptation is cross-fitted. Shared-context batches require group-level folds and disclosure.
  \item \textbf{Nuisance and weight regularity.} For the response-model point-consistency arm, the cross-fitted propensity estimator satisfies $\|\widehat\pi-\pi\|_{2,w}=o_p(1)$ and the fold-independent outcome predictions are uniformly bounded; no outcome-model consistency is assumed. For the model-assisted arm, $\widehat\mu\to\mu$ while the propensity limit need only remain bounded away from zero. For the root-$M$ expansion, influence-function form, and normal-interval claims, both nuisance estimators must be consistent with the product rate
  \[
  \|\widehat\pi-\pi\|_{2,w}
  \|\widehat\mu-\mu\|_{2,w}
  =o_p(M^{-1/2}).
  \]
  In addition, $A_M/M\to\bar a>0$, $M^{-1}\sum_i a_i^2=O(1)$, and the normalized weights and score satisfy the corresponding Lindeberg and finite-moment conditions.
  \item \textbf{Superpopulation model.} Frame units are iid, or iid clusters when a cluster design is declared, drawn from a superpopulation; cluster statements use cluster-robust inference. This sampling law is required for every superpopulation-labeled statement: the model-assisted point-consistency arm, the superpopulation influence function and sandwich, and any DML or efficiency claim. It is not invoked for finite-frame response-design statements.
\end{enumerate}
\endgroup

\subsection{Frame-AIPW Point Robustness}

\begingroup
For completeness, we restate the estimator so this appendix is self-contained. Let $A_M=\sum_{i\in U}a_i$, and let $k(i)$ denote the held-out fold containing frame unit $i$:
\begin{equation}
\widehat\theta_U
=
\frac{1}{A_M}
\sum_{i\in U}a_i
\left[
\widehat\mu_{-k(i)}(W_i)
+
\frac{R_i}{\widehat\pi_{-k(i)}(W_i)}
\{Y_i-\widehat\mu_{-k(i)}(W_i)\}
\right].
\label{eq:app_frame_aipw}
\end{equation}

For candidate nuisance limits $q$ and $h$, define
\[
m(q,h)=h(W)+\frac{R}{q(W)}\{Y-h(W)\}.
\]
Under MAR, conditional on the nuisance-training sample $\mathcal T$,
\begin{align}
\E[m(q,h)\mid W,\mathcal T]
&=
h(W)+\frac{\pi(W)}{q(W)}\{\mu(W)-h(W)\},\\
\E[m(q,h)-Y\mid W,\mathcal T]
&=
\left\{1-\frac{\pi(W)}{q(W)}\right\}
\{h(W)-\mu(W)\}.
\label{eq:app_product_bias}
\end{align}
The two arms carry different assumptions, and neither requires both nuisance models to be consistent:
\begin{enumerate}[label=(\alph*)]
  \item \textbf{Response-model arm.} If $q=\pi$ exactly, the right side of \eqref{eq:app_product_bias} is zero for every fold-independent $h$, whatever its consistency. With estimated propensities, consistency of $\widehat\pi$ for $\pi$, uniformly bounded fold-independent outcome predictions, and the weight regularity conditions yield consistency for $\theta_U$. No outcome-model consistency is used.
  \item \textbf{Outcome-model arm.} Under the superpopulation sampling model, if $h=\mu$ while $q$ may be misspecified but stays bounded away from zero, the weighted law of large numbers gives consistency for $\theta_w=\E[aY]/\E[a]$. No propensity-model consistency is used. This is not an exact finite-frame unbiasedness statement for a realized frame.
\end{enumerate}

For the root-$M$ expansion, influence-function form, and normal intervals, both nuisance estimators must in addition satisfy the dual-consistency product rate above. These are point-consistency and inference-scope statements; one-model-correct point consistency alone does not establish one-model-correct confidence intervals.
\endgroup

\subsection{Finite-Frame Expansion and Response-Design Variance}

\begingroup
Under the finite-frame response design and the stated nuisance-rate and weight regularity conditions,
\begin{equation}
\widehat\theta_U-\theta_U
=
\frac{1}{A_M}
\sum_{i\in U}
a_i\left(\frac{R_i}{\pi_i}-1\right)
\{Y_i-\mu_i\}
+o_p(M^{-1/2}).
\label{eq:app_fp_expansion}
\end{equation}
Conditional on the realized frame and independent Bernoulli response,
\begin{equation}
V_R
=
\frac{1}{A_M^2}
\sum_{i\in U}
a_i^2\frac{1-\pi_i}{\pi_i}
\{Y_i-\mu_i\}^2.
\label{eq:app_fp_variance}
\end{equation}
The cross-fitted estimator is
\begin{equation}
\widehat V_R
=
\frac{1}{A_M^2}
\sum_{i\in U}
a_i^2R_i
\frac{1-\widehat\pi_i}{\widehat\pi_i^2}
\{Y_i-\widehat\mu_i\}^2.
\label{eq:app_fp_variance_estimator}
\end{equation}
If every frame unit responds, then $R_i=\pi_i=1$, $\widehat\theta_U=\theta_U$ exactly, and the response-design variance is zero.
\endgroup

\subsection{Superpopulation Influence Function and Sandwich}

\begingroup
For the secondary target $\theta_w=\E[aY]/\E[a]$, the influence function under the declared superpopulation model is
\begin{equation}
\phi(O)
=
\frac{a}{\E[a]}
\left[
\mu(W)-\theta_w
+
\frac{R}{\pi(W)}
\{Y-\mu(W)\}
\right].
\label{eq:app_sp_if}
\end{equation}
With cross-fitted pseudo-outcomes $\widehat m_i$, the centered sandwich estimator is
\begin{equation}
\widehat V_{\mathrm{SP}}
=
\frac{M}{M-1}
\frac{1}{A_M^2}
\sum_{i\in U}
a_i^2\{\widehat m_i-\widehat\theta_U\}^2.
\label{eq:app_sp_variance_estimator}
\end{equation}
\endgroup

\subsection{Conditions for Normal Intervals}

\begingroup
A normal interval is justified only under the declared target reading, the corresponding response or sampling design, positivity, moment and weight regularity, cross-fitting, and nuisance convergence with the required product rate. One-model-correct point consistency is not sufficient. For the finite-frame reading, the standard error is $\sqrt{\widehat V_R}$; for the superpopulation reading, it is $\sqrt{\widehat V_{\mathrm{SP}}}$. Fold-to-fold variation in a small number of fold estimates is not a substitute for either variance estimator.
\endgroup

\section{Supplementary Diagnostics}
\label{app:supplementary_diagnostics}

\subsection{Diagnostic 4: Marginal Calibration}
\label{app:diag_marginal}

\paragraph{Question.} Do synthetic marginal response distributions match available benchmarks?

\paragraph{Estimation.} For each outcome variable $Y_k$, compute the Kolmogorov--Smirnov statistic between the synthetic marginal distribution $\hat{F}_{\cS,k}$ and the benchmark distribution $F_{\cT,k}$ (from real data or published statistics):
\begin{equation}
  D_k = \sup_{y} \bigl|\hat{F}_{\cS,k}(y) - F_{\cT,k}(y)\bigr|.
  \label{eq:ks_stat}
\end{equation}

\paragraph{Decision rule.}
\begin{itemize}[leftmargin=2em]
  \item $D_k < D_{\alpha}^{*}$ (KS critical value at level $\alpha$): Marginal distributions are statistically compatible.
  \item $D_k \geq D_{\alpha}^{*}$: Marginal distributions differ. This is necessary but not sufficient evidence of concept shift---marginal differences can arise from covariate shift alone.
\end{itemize}

\paragraph{Limitation.} Marginal calibration is a weak test: marginal distributions can match even when conditional distributions diverge substantially. It is most useful as a sanity check. If marginals do not match, something is clearly wrong. If they do match, the practitioner should not be reassured---conditional calibration is the relevant test.

\begingroup
 Suppose a synthetic panel and a real benchmark both show 40\% purchase intent overall. The panel can still be wrong where the business decision is made. A 10-point overstatement among younger customers and a 10-point understatement among an equally large group of older customers cancel in the topline. A brand allocating budgets by segment would still base both allocations on the wrong numbers, so marginal agreement is only the first check.
\endgroup

\subsection{Diagnostic 5: Sensitivity to Unobserved Concept Shift}
\label{app:diag_sensitivity}

\paragraph{Question.} How large would unobserved concept shift need to be to reverse conclusions?

\paragraph{Estimation.} Adapting Rosenbaum-style sensitivity bounds to the concept shift setting, we parameterize the maximum concept shift by a sensitivity parameter $\Gamma > 0$:
\begin{equation}
  \sup_{x \in \cX} \bigl|\muT(x) - \muS(x)\bigr| \leq \Gamma.
  \label{eq:sensitivity_bound}
\end{equation}

For each value of $\Gamma$, compute the worst-case estimate $\tauDR^{+}(\Gamma)$ and $\tauDR^{-}(\Gamma)$ that are consistent with the data and the bound $\Gamma$. The conclusion is robust if $\sign(\tauDR^{+}(\Gamma)) = \sign(\tauDR^{-}(\Gamma))$ for all $\Gamma$ up to a plausible maximum.

\paragraph{Connection to Jin et al.} \citet{jin2025beyondreweighting} provide an empirical bound: the magnitude of concept shift is typically bounded by the magnitude of covariate shift. This informs the choice of ``plausible maximum'' $\Gamma$: if covariate shift as measured by Diagnostic~1 implies moderate distributional divergence, the sensitivity analysis explores $\Gamma$ values up to that magnitude.

\paragraph{Decision rule.}
\begin{itemize}[leftmargin=2em]
  \item Conclusions hold for $\Gamma$ up to the Jin et al.\ empirical bound: \textbf{Robust.} Proceed with DR correction.
  \item Conclusions reverse for $\Gamma$ within the plausible range: \textbf{Fragile. Walk Away.}
\end{itemize}

\begingroup
Suppose a promotion is expected to lift purchase by three percentage points, but plausible concept shift could move the estimate by five points in either direction. The worst-case effect is a two-point decline. The synthetic result therefore does not support a stable rollout decision; the firm should collect more real responses or defer the promotion.
\endgroup


\begingroup
\section{Final Simulation and Controlled-Masking Results}
\label{app:final_results}

This appendix reports the complete seven-estimator summaries from the frozen final artifacts. The battery contains 20 DGP--nuisance cells (36,000 replication-cells; 252,000 estimator rows). The controlled-masking experiment contains 12 dataset--outcome--response cells (12,000 replication-cells; 84,000 estimator rows). Coverage is calculated over all attempts; for frame-AIPW, an estimate withheld by the positivity gate counts as noncoverage. Bias and RMSE use reportable estimates. A coverage interval wholly below 95\% is a failure; an interval wholly above 95\% is conservative; otherwise it is compatible with 95\%.

\subsection{Frozen Simulation Battery}
\tiny
\setlength{\tabcolsep}{2.5pt}
\begin{longtable}{p{2.0cm}llp{1.6cm}rrp{2.4cm}p{1.5cm}}
\caption{Complete seven-estimator simulation battery. Bias and RMSE are in each DGP outcome scale; MC is the Wilson 95\% Monte Carlo interval for coverage.}\label{app:tab:final_battery}\\
\toprule
\textbf{DGP} & \textbf{$\pi$} & \textbf{$\mu$} & \textbf{Estimator} & \textbf{Bias} & \textbf{RMSE} & \textbf{Coverage [95\% MC]} & \textbf{Status} \\
\midrule
\endfirsthead
\multicolumn{8}{c}{\tablename\ \thetable\ continued}\\
\toprule
\textbf{DGP} & \textbf{$\pi$} & \textbf{$\mu$} & \textbf{Estimator} & \textbf{Bias} & \textbf{RMSE} & \textbf{Coverage [95\% MC]} & \textbf{Status} \\
\midrule
\endhead
\bottomrule
\endfoot
Probability sampling & correct & correct & Naive synthetic & 0.0011 & 0.0310 & 99.65\% [99.28, 99.83] & conservative \\
Probability sampling & correct & correct & H\'ajek IPW & $-0.0087$ & 0.0926 & 99.80\% [99.49, 99.92] & conservative \\
Probability sampling & correct & correct & Plain AIPW & 0.0005 & 0.0473 & 93.90\% [92.76, 94.87] & \textbf{failure} \\
Probability sampling & correct & correct & PPI & 0.0011 & 0.0339 & 100.00\% [99.81, 100.00] & conservative \\
Probability sampling & correct & correct & PPI++ & 0.0943 & 0.1016 & 73.05\% [71.06, 74.95] & \textbf{failure} \\
Probability sampling & correct & correct & Legacy & 0.2707 & 0.2769 & 91.65\% [90.36, 92.78] & \textbf{failure} \\
Probability sampling & correct & correct & Frame-AIPW & 0.0009 & 0.0502 & 95.15\% [94.12, 96.01] & compatible \\
Probability sampling & correct & wrong & Naive synthetic & 0.0011 & 0.0310 & 99.65\% [99.28, 99.83] & conservative \\
Probability sampling & correct & wrong & H\'ajek IPW & $-0.0087$ & 0.0926 & 99.80\% [99.49, 99.92] & conservative \\
Probability sampling & correct & wrong & Plain AIPW & 0.0022 & 0.0862 & 99.00\% [98.46, 99.35] & conservative \\
Probability sampling & correct & wrong & PPI & 0.0011 & 0.0339 & 100.00\% [99.81, 100.00] & conservative \\
Probability sampling & correct & wrong & PPI++ & 0.0943 & 0.1016 & 73.05\% [71.06, 74.95] & \textbf{failure} \\
Probability sampling & correct & wrong & Legacy & 0.2707 & 0.2769 & 91.65\% [90.36, 92.78] & \textbf{failure} \\
Probability sampling & correct & wrong & Frame-AIPW & $-0.0144$ & 0.0990 & 99.35\% [98.89, 99.62] & conservative \\
Probability sampling & wrong & correct & Naive synthetic & 0.0011 & 0.0310 & 99.65\% [99.28, 99.83] & conservative \\
Probability sampling & wrong & correct & H\'ajek IPW & 0.2704 & 0.2770 & 2.50\% [1.90, 3.28] & \textbf{failure} \\
Probability sampling & wrong & correct & Plain AIPW & 0.0003 & 0.0411 & 87.40\% [85.87, 88.78] & \textbf{failure} \\
Probability sampling & wrong & correct & PPI & 0.0011 & 0.0339 & 100.00\% [99.81, 100.00] & conservative \\
Probability sampling & wrong & correct & PPI++ & 0.0943 & 0.1016 & 73.05\% [71.06, 74.95] & \textbf{failure} \\
Probability sampling & wrong & correct & Legacy & 0.2707 & 0.2769 & 91.65\% [90.36, 92.78] & \textbf{failure} \\
Probability sampling & wrong & correct & Frame-AIPW & 0.0004 & 0.0411 & 88.15\% [86.66, 89.49] & \textbf{failure} \\
Probability sampling & wrong & wrong & Naive synthetic & 0.0011 & 0.0310 & 99.65\% [99.28, 99.83] & conservative \\
Probability sampling & wrong & wrong & H\'ajek IPW & 0.2704 & 0.2770 & 2.50\% [1.90, 3.28] & \textbf{failure} \\
Probability sampling & wrong & wrong & Plain AIPW & 0.2705 & 0.2771 & 0.35\% [0.17, 0.72] & \textbf{failure} \\
Probability sampling & wrong & wrong & PPI & 0.0011 & 0.0339 & 100.00\% [99.81, 100.00] & conservative \\
Probability sampling & wrong & wrong & PPI++ & 0.0943 & 0.1016 & 73.05\% [71.06, 74.95] & \textbf{failure} \\
Probability sampling & wrong & wrong & Legacy & 0.2707 & 0.2769 & 91.65\% [90.36, 92.78] & \textbf{failure} \\
Probability sampling & wrong & wrong & Frame-AIPW & 0.2706 & 0.2771 & 0.35\% [0.17, 0.72] & \textbf{failure} \\
MAR missingness & correct & correct & Naive synthetic & $-0.0004$ & 0.0311 & 99.85\% [99.56, 99.95] & conservative \\
MAR missingness & correct & correct & H\'ajek IPW & $-0.0019$ & 0.0422 & 99.95\% [99.72, 99.99] & conservative \\
MAR missingness & correct & correct & Plain AIPW & 0.0001 & 0.0347 & 94.10\% [92.98, 95.05] & compatible \\
MAR missingness & correct & correct & PPI & $-0.0003$ & 0.0320 & 100.00\% [99.81, 100.00] & conservative \\
MAR missingness & correct & correct & PPI++ & 0.0210 & 0.0413 & 99.60\% [99.21, 99.80] & conservative \\
MAR missingness & correct & correct & Legacy & 0.0598 & 0.0825 & 100.00\% [99.81, 100.00] & conservative \\
MAR missingness & correct & correct & Frame-AIPW & 0.0003 & 0.0354 & 95.10\% [94.06, 95.96] & compatible \\
MAR missingness & correct & wrong & Naive synthetic & $-0.0004$ & 0.0311 & 99.85\% [99.56, 99.95] & conservative \\
MAR missingness & correct & wrong & H\'ajek IPW & $-0.0019$ & 0.0422 & 99.95\% [99.72, 99.99] & conservative \\
MAR missingness & correct & wrong & Plain AIPW & 0.0001 & 0.0408 & 99.80\% [99.49, 99.92] & conservative \\
MAR missingness & correct & wrong & PPI & $-0.0003$ & 0.0320 & 100.00\% [99.81, 100.00] & conservative \\
MAR missingness & correct & wrong & PPI++ & 0.0210 & 0.0413 & 99.60\% [99.21, 99.80] & conservative \\
MAR missingness & correct & wrong & Legacy & 0.0598 & 0.0825 & 100.00\% [99.81, 100.00] & conservative \\
MAR missingness & correct & wrong & Frame-AIPW & $-0.0026$ & 0.0431 & 99.85\% [99.56, 99.95] & conservative \\
MAR missingness & wrong & correct & Naive synthetic & $-0.0004$ & 0.0311 & 99.85\% [99.56, 99.95] & conservative \\
MAR missingness & wrong & correct & H\'ajek IPW & 0.0598 & 0.0826 & 94.80\% [93.74, 95.69] & compatible \\
MAR missingness & wrong & correct & Plain AIPW & 0.0003 & 0.0340 & 91.75\% [90.46, 92.88] & \textbf{failure} \\
MAR missingness & wrong & correct & PPI & $-0.0003$ & 0.0320 & 100.00\% [99.81, 100.00] & conservative \\
MAR missingness & wrong & correct & PPI++ & 0.0210 & 0.0413 & 99.60\% [99.21, 99.80] & conservative \\
MAR missingness & wrong & correct & Legacy & 0.0598 & 0.0825 & 100.00\% [99.81, 100.00] & conservative \\
MAR missingness & wrong & correct & Frame-AIPW & 0.0003 & 0.0341 & 92.55\% [91.32, 93.62] & \textbf{failure} \\
MAR missingness & wrong & wrong & Naive synthetic & $-0.0004$ & 0.0311 & 99.85\% [99.56, 99.95] & conservative \\
MAR missingness & wrong & wrong & H\'ajek IPW & 0.0598 & 0.0826 & 94.80\% [93.74, 95.69] & compatible \\
MAR missingness & wrong & wrong & Plain AIPW & 0.0599 & 0.0826 & 80.85\% [79.07, 82.51] & \textbf{failure} \\
MAR missingness & wrong & wrong & PPI & $-0.0003$ & 0.0320 & 100.00\% [99.81, 100.00] & conservative \\
MAR missingness & wrong & wrong & PPI++ & 0.0210 & 0.0413 & 99.60\% [99.21, 99.80] & conservative \\
MAR missingness & wrong & wrong & Legacy & 0.0598 & 0.0825 & 100.00\% [99.81, 100.00] & conservative \\
MAR missingness & wrong & wrong & Frame-AIPW & 0.0599 & 0.0827 & 81.30\% [79.53, 82.95] & \textbf{failure} \\
Covariate shift & correct & correct & Naive synthetic & $-0.0008$ & 0.0314 & 99.75\% [99.42, 99.89] & conservative \\
Covariate shift & correct & correct & H\'ajek IPW & $-0.0071$ & 0.0875 & 99.70\% [99.35, 99.86] & conservative \\
Covariate shift & correct & correct & Plain AIPW & 0.0001 & 0.0463 & 95.25\% [94.23, 96.10] & compatible \\
Covariate shift & correct & correct & PPI & 0.0002 & 0.0341 & 100.00\% [99.81, 100.00] & conservative \\
Covariate shift & correct & correct & PPI++ & 0.0923 & 0.0995 & 72.95\% [70.96, 74.85] & \textbf{failure} \\
Covariate shift & correct & correct & Legacy & 0.2675 & 0.2818 & 98.60\% [97.98, 99.03] & conservative \\
Covariate shift & correct & correct & Frame-AIPW & 0.0002 & 0.0489 & 95.65\% [94.67, 96.46] & compatible \\
Covariate shift & correct & wrong & Naive synthetic & $-0.0008$ & 0.0314 & 99.75\% [99.42, 99.89] & conservative \\
Covariate shift & correct & wrong & H\'ajek IPW & $-0.0071$ & 0.0875 & 99.70\% [99.35, 99.86] & conservative \\
Covariate shift & correct & wrong & Plain AIPW & 0.0044 & 0.0873 & 99.10\% [98.58, 99.43] & conservative \\
Covariate shift & correct & wrong & PPI & 0.0002 & 0.0341 & 100.00\% [99.81, 100.00] & conservative \\
Covariate shift & correct & wrong & PPI++ & 0.0923 & 0.0995 & 72.95\% [70.96, 74.85] & \textbf{failure} \\
Covariate shift & correct & wrong & Legacy & 0.2675 & 0.2818 & 98.60\% [97.98, 99.03] & conservative \\
Covariate shift & correct & wrong & Frame-AIPW & $-0.0297$ & 0.1040 & 99.50\% [99.08, 99.73] & conservative \\
Covariate shift & wrong & correct & Naive synthetic & $-0.0008$ & 0.0314 & 99.75\% [99.42, 99.89] & conservative \\
Covariate shift & wrong & correct & H\'ajek IPW & 0.2666 & 0.2730 & 2.90\% [2.25, 3.73] & \textbf{failure} \\
Covariate shift & wrong & correct & Plain AIPW & 0.0001 & 0.0339 & 94.20\% [93.09, 95.14] & compatible \\
Covariate shift & wrong & correct & PPI & 0.0002 & 0.0341 & 100.00\% [99.81, 100.00] & conservative \\
Covariate shift & wrong & correct & PPI++ & 0.0923 & 0.0995 & 72.95\% [70.96, 74.85] & \textbf{failure} \\
Covariate shift & wrong & correct & Legacy & 0.2675 & 0.2818 & 98.60\% [97.98, 99.03] & conservative \\
Covariate shift & wrong & correct & Frame-AIPW & 0.0001 & 0.0340 & 94.15\% [93.03, 95.10] & compatible \\
Covariate shift & wrong & wrong & Naive synthetic & $-0.0008$ & 0.0314 & 99.75\% [99.42, 99.89] & conservative \\
Covariate shift & wrong & wrong & H\'ajek IPW & 0.2666 & 0.2730 & 2.90\% [2.25, 3.73] & \textbf{failure} \\
Covariate shift & wrong & wrong & Plain AIPW & 0.2666 & 0.2730 & 0.75\% [0.46, 1.23] & \textbf{failure} \\
Covariate shift & wrong & wrong & PPI & 0.0002 & 0.0341 & 100.00\% [99.81, 100.00] & conservative \\
Covariate shift & wrong & wrong & PPI++ & 0.0923 & 0.0995 & 72.95\% [70.96, 74.85] & \textbf{failure} \\
Covariate shift & wrong & wrong & Legacy & 0.2675 & 0.2818 & 98.60\% [97.98, 99.03] & conservative \\
Covariate shift & wrong & wrong & Frame-AIPW & 0.2647 & 0.2712 & 1.05\% [0.69, 1.60] & \textbf{failure} \\
Concept shift & correct & correct & Naive synthetic & 0.1804 & 0.1844 & 7.20\% [6.15, 8.42] & \textbf{failure} \\
Concept shift & correct & correct & H\'ajek IPW & $-0.0059$ & 0.0328 & 100.00\% [99.81, 100.00] & conservative \\
Concept shift & correct & correct & Plain AIPW & $-0.0001$ & 0.0256 & 94.05\% [92.93, 95.00] & compatible \\
Concept shift & correct & correct & PPI & $-0.1062$ & 0.1105 & 94.30\% [93.20, 95.23] & compatible \\
Concept shift & correct & correct & PPI++ & 0.0746 & 0.0809 & 85.45\% [83.84, 86.93] & \textbf{failure} \\
Concept shift & correct & correct & Legacy & 0.2169 & 0.2215 & 98.55\% [97.93, 98.99] & conservative \\
Concept shift & correct & correct & Frame-AIPW & $-0.0000$ & 0.0258 & 95.05\% [94.01, 95.92] & compatible \\
Concept shift & correct & wrong & Naive synthetic & 0.1804 & 0.1844 & 7.20\% [6.15, 8.42] & \textbf{failure} \\
Concept shift & correct & wrong & H\'ajek IPW & $-0.0059$ & 0.0328 & 100.00\% [99.81, 100.00] & conservative \\
Concept shift & correct & wrong & Plain AIPW & 0.0000 & 0.0314 & 100.00\% [99.81, 100.00] & conservative \\
Concept shift & correct & wrong & PPI & $-0.1062$ & 0.1105 & 94.30\% [93.20, 95.23] & compatible \\
Concept shift & correct & wrong & PPI++ & 0.0746 & 0.0809 & 85.45\% [83.84, 86.93] & \textbf{failure} \\
Concept shift & correct & wrong & Legacy & 0.2169 & 0.2215 & 98.55\% [97.93, 98.99] & conservative \\
Concept shift & correct & wrong & Frame-AIPW & $-0.0080$ & 0.0337 & 99.90\% [99.64, 99.97] & conservative \\
Concept shift & wrong & correct & Naive synthetic & 0.1804 & 0.1844 & 7.20\% [6.15, 8.42] & \textbf{failure} \\
Concept shift & wrong & correct & H\'ajek IPW & 0.2170 & 0.2218 & 4.25\% [3.45, 5.23] & \textbf{failure} \\
Concept shift & wrong & correct & Plain AIPW & $-0.0000$ & 0.0254 & 92.90\% [91.69, 93.95] & \textbf{failure} \\
Concept shift & wrong & correct & PPI & $-0.1062$ & 0.1105 & 94.30\% [93.20, 95.23] & compatible \\
Concept shift & wrong & correct & PPI++ & 0.0746 & 0.0809 & 85.45\% [83.84, 86.93] & \textbf{failure} \\
Concept shift & wrong & correct & Legacy & 0.2169 & 0.2215 & 98.55\% [97.93, 98.99] & conservative \\
Concept shift & wrong & correct & Frame-AIPW & 0.0000 & 0.0255 & 93.45\% [92.28, 94.45] & \textbf{failure} \\
Concept shift & wrong & wrong & Naive synthetic & 0.1804 & 0.1844 & 7.20\% [6.15, 8.42] & \textbf{failure} \\
Concept shift & wrong & wrong & H\'ajek IPW & 0.2170 & 0.2218 & 4.25\% [3.45, 5.23] & \textbf{failure} \\
Concept shift & wrong & wrong & Plain AIPW & 0.2170 & 0.2217 & 0.10\% [0.03, 0.36] & \textbf{failure} \\
Concept shift & wrong & wrong & PPI & $-0.1062$ & 0.1105 & 94.30\% [93.20, 95.23] & compatible \\
Concept shift & wrong & wrong & PPI++ & 0.0746 & 0.0809 & 85.45\% [83.84, 86.93] & \textbf{failure} \\
Concept shift & wrong & wrong & Legacy & 0.2169 & 0.2215 & 98.55\% [97.93, 98.99] & conservative \\
Concept shift & wrong & wrong & Frame-AIPW & 0.2171 & 0.2218 & 0.10\% [0.03, 0.36] & \textbf{failure} \\
MNAR self-selection & correct & correct & Naive synthetic & $-0.0004$ & 0.0323 & 99.70\% [99.12, 99.90] & conservative \\
MNAR self-selection & correct & correct & H\'ajek IPW & 0.2336 & 0.2354 & 0.20\% [0.05, 0.73] & \textbf{failure} \\
MNAR self-selection & correct & correct & Plain AIPW & 0.2336 & 0.2353 & 0.00\% [0.00, 0.38] & \textbf{failure} \\
MNAR self-selection & correct & correct & PPI & 0.2333 & 0.2350 & 0.00\% [0.00, 0.38] & \textbf{failure} \\
MNAR self-selection & correct & correct & PPI++ & 0.2331 & 0.2352 & 0.00\% [0.00, 0.38] & \textbf{failure} \\
MNAR self-selection & correct & correct & Legacy & 0.2327 & 0.2377 & 98.10\% [97.05, 98.78] & conservative \\
MNAR self-selection & correct & correct & Frame-AIPW & 0.2338 & 0.2355 & 0.00\% [0.00, 0.38] & \textbf{failure} \\
MNAR self-selection & correct & wrong & Naive synthetic & $-0.0004$ & 0.0323 & 99.70\% [99.12, 99.90] & conservative \\
MNAR self-selection & correct & wrong & H\'ajek IPW & 0.2336 & 0.2354 & 0.20\% [0.05, 0.73] & \textbf{failure} \\
MNAR self-selection & correct & wrong & Plain AIPW & 0.2335 & 0.2352 & 0.00\% [0.00, 0.38] & \textbf{failure} \\
MNAR self-selection & correct & wrong & PPI & 0.2333 & 0.2350 & 0.00\% [0.00, 0.38] & \textbf{failure} \\
MNAR self-selection & correct & wrong & PPI++ & 0.2331 & 0.2352 & 0.00\% [0.00, 0.38] & \textbf{failure} \\
MNAR self-selection & correct & wrong & Legacy & 0.2327 & 0.2377 & 98.10\% [97.05, 98.78] & conservative \\
MNAR self-selection & correct & wrong & Frame-AIPW & 0.2338 & 0.2356 & 0.00\% [0.00, 0.38] & \textbf{failure} \\
MNAR self-selection & wrong & correct & Naive synthetic & $-0.0004$ & 0.0323 & 99.70\% [99.12, 99.90] & conservative \\
MNAR self-selection & wrong & correct & H\'ajek IPW & 0.2325 & 0.2375 & 3.50\% [2.53, 4.83] & \textbf{failure} \\
MNAR self-selection & wrong & correct & Plain AIPW & 0.2336 & 0.2353 & 0.00\% [0.00, 0.38] & \textbf{failure} \\
MNAR self-selection & wrong & correct & PPI & 0.2333 & 0.2350 & 0.00\% [0.00, 0.38] & \textbf{failure} \\
MNAR self-selection & wrong & correct & PPI++ & 0.2331 & 0.2352 & 0.00\% [0.00, 0.38] & \textbf{failure} \\
MNAR self-selection & wrong & correct & Legacy & 0.2327 & 0.2377 & 98.10\% [97.05, 98.78] & conservative \\
MNAR self-selection & wrong & correct & Frame-AIPW & 0.2337 & 0.2354 & 0.00\% [0.00, 0.38] & \textbf{failure} \\
MNAR self-selection & wrong & wrong & Naive synthetic & $-0.0004$ & 0.0323 & 99.70\% [99.12, 99.90] & conservative \\
MNAR self-selection & wrong & wrong & H\'ajek IPW & 0.2325 & 0.2375 & 3.50\% [2.53, 4.83] & \textbf{failure} \\
MNAR self-selection & wrong & wrong & Plain AIPW & 0.2326 & 0.2376 & 0.20\% [0.05, 0.73] & \textbf{failure} \\
MNAR self-selection & wrong & wrong & PPI & 0.2333 & 0.2350 & 0.00\% [0.00, 0.38] & \textbf{failure} \\
MNAR self-selection & wrong & wrong & PPI++ & 0.2331 & 0.2352 & 0.00\% [0.00, 0.38] & \textbf{failure} \\
MNAR self-selection & wrong & wrong & Legacy & 0.2327 & 0.2377 & 98.10\% [97.05, 98.78] & conservative \\
MNAR self-selection & wrong & wrong & Frame-AIPW & 0.2328 & 0.2378 & 0.20\% [0.05, 0.73] & \textbf{failure} \\
\end{longtable}
\normalsize

\subsection{Controlled Masking on Paired Data}
\tiny
\setlength{\tabcolsep}{2.5pt}
\begin{longtable}{lllp{1.6cm}rrrp{2.4cm}p{1.5cm}}
\caption{Complete seven-estimator controlled-masking results. ``Avail.'' is reportable estimates over 1,000 attempts. Coverage includes unavailable frame-AIPW attempts as noncoverage.}\label{app:tab:final_masking}\\
\toprule
\textbf{Data} & \textbf{Outcome} & \textbf{Regime} & \textbf{Estimator} & \textbf{Avail.} & \textbf{Bias} & \textbf{RMSE} & \textbf{Coverage [95\% MC]} & \textbf{Status} \\
\midrule
\endfirsthead
\multicolumn{9}{c}{\tablename\ \thetable\ continued}\\
\toprule
\textbf{Data} & \textbf{Outcome} & \textbf{Regime} & \textbf{Estimator} & \textbf{Avail.} & \textbf{Bias} & \textbf{RMSE} & \textbf{Coverage [95\% MC]} & \textbf{Status} \\
\midrule
\endhead
\bottomrule
\endfoot
ANES & Dem. thermometer & Mild MAR & Naive synthetic & 1,000/1,000 & 12.0796 & 12.0796 & 0.00\% [0.00, 0.38] & \textbf{failure} \\
ANES & Dem. thermometer & Mild MAR & H\'ajek IPW & 1,000/1,000 & 0.0611 & 0.4215 & 100.00\% [99.62, 100.00] & conservative \\
ANES & Dem. thermometer & Mild MAR & Plain AIPW & 1,000/1,000 & 0.0039 & 0.3962 & 95.30\% [93.81, 96.45] & compatible \\
ANES & Dem. thermometer & Mild MAR & PPI & 1,000/1,000 & 0.0257 & 0.3672 & 100.00\% [99.62, 100.00] & conservative \\
ANES & Dem. thermometer & Mild MAR & PPI++ & 1,000/1,000 & $-0.2573$ & 0.4707 & 99.80\% [99.27, 99.95] & conservative \\
ANES & Dem. thermometer & Mild MAR & Legacy & 1,000/1,000 & $-0.7310$ & 0.9159 & 100.00\% [99.62, 100.00] & conservative \\
ANES & Dem. thermometer & Mild MAR & Frame-AIPW & 1,000/1,000 & 0.0149 & 0.4116 & 95.80\% [94.37, 96.88] & compatible \\
ANES & Dem. thermometer & Strong MAR & Naive synthetic & 1,000/1,000 & 12.0796 & 12.0796 & 0.00\% [0.00, 0.38] & \textbf{failure} \\
ANES & Dem. thermometer & Strong MAR & H\'ajek IPW & 1,000/1,000 & 0.1065 & 0.9578 & 99.90\% [99.44, 99.98] & conservative \\
ANES & Dem. thermometer & Strong MAR & Plain AIPW & 1,000/1,000 & $-0.0216$ & 0.6925 & 94.10\% [92.46, 95.40] & compatible \\
ANES & Dem. thermometer & Strong MAR & PPI & 1,000/1,000 & $-0.0226$ & 0.5244 & 99.80\% [99.27, 99.95] & conservative \\
ANES & Dem. thermometer & Strong MAR & PPI++ & 1,000/1,000 & $-0.5351$ & 0.7638 & 97.20\% [95.98, 98.06] & conservative \\
ANES & Dem. thermometer & Strong MAR & Legacy & 1,000/1,000 & $-1.7563$ & 1.9200 & 100.00\% [99.62, 100.00] & conservative \\
ANES & Dem. thermometer & Strong MAR & Frame-AIPW & 992/1,000 & 0.0041 & 0.7906 & 94.90\% [93.36, 96.10] & compatible \\
ANES & Dem. thermometer & MNAR & Naive synthetic & 1,000/1,000 & 12.0796 & 12.0796 & 0.00\% [0.00, 0.38] & \textbf{failure} \\
ANES & Dem. thermometer & MNAR & H\'ajek IPW & 1,000/1,000 & 5.1184 & 5.1518 & 0.00\% [0.00, 0.38] & \textbf{failure} \\
ANES & Dem. thermometer & MNAR & Plain AIPW & 1,000/1,000 & 5.2549 & 5.2790 & 0.00\% [0.00, 0.38] & \textbf{failure} \\
ANES & Dem. thermometer & MNAR & PPI & 1,000/1,000 & 5.0419 & 5.0617 & 0.00\% [0.00, 0.38] & \textbf{failure} \\
ANES & Dem. thermometer & MNAR & PPI++ & 1,000/1,000 & 7.4542 & 7.4719 & 0.00\% [0.00, 0.38] & \textbf{failure} \\
ANES & Dem. thermometer & MNAR & Legacy & 1,000/1,000 & 11.5905 & 11.6109 & 0.00\% [0.00, 0.38] & \textbf{failure} \\
ANES & Dem. thermometer & MNAR & Frame-AIPW & 1,000/1,000 & 5.2855 & 5.3115 & 0.00\% [0.00, 0.38] & \textbf{failure} \\
ANES & Rep. thermometer & Mild MAR & Naive synthetic & 1,000/1,000 & 13.5457 & 13.5457 & 0.00\% [0.00, 0.38] & \textbf{failure} \\
ANES & Rep. thermometer & Mild MAR & H\'ajek IPW & 1,000/1,000 & $-0.0415$ & 0.4136 & 100.00\% [99.62, 100.00] & conservative \\
ANES & Rep. thermometer & Mild MAR & Plain AIPW & 1,000/1,000 & 0.0049 & 0.3959 & 95.40\% [93.92, 96.53] & compatible \\
ANES & Rep. thermometer & Mild MAR & PPI & 1,000/1,000 & $-0.3555$ & 0.5120 & 99.90\% [99.44, 99.98] & conservative \\
ANES & Rep. thermometer & Mild MAR & PPI++ & 1,000/1,000 & $-0.0813$ & 0.3964 & 100.00\% [99.62, 100.00] & conservative \\
ANES & Rep. thermometer & Mild MAR & Legacy & 1,000/1,000 & 0.2662 & 0.5672 & 100.00\% [99.62, 100.00] & conservative \\
ANES & Rep. thermometer & Mild MAR & Frame-AIPW & 1,000/1,000 & 0.0061 & 0.4101 & 95.70\% [94.26, 96.79] & compatible \\
ANES & Rep. thermometer & Strong MAR & Naive synthetic & 1,000/1,000 & 13.5457 & 13.5457 & 0.00\% [0.00, 0.38] & \textbf{failure} \\
ANES & Rep. thermometer & Strong MAR & H\'ajek IPW & 1,000/1,000 & $-0.1685$ & 1.1676 & 99.30\% [98.56, 99.66] & conservative \\
ANES & Rep. thermometer & Strong MAR & Plain AIPW & 1,000/1,000 & $-0.0266$ & 0.6606 & 95.30\% [93.81, 96.45] & compatible \\
ANES & Rep. thermometer & Strong MAR & PPI & 1,000/1,000 & $-0.9263$ & 1.0579 & 93.70\% [92.02, 95.05] & compatible \\
ANES & Rep. thermometer & Strong MAR & PPI++ & 1,000/1,000 & $-0.4243$ & 0.6660 & 98.70\% [97.79, 99.24] & conservative \\
ANES & Rep. thermometer & Strong MAR & Legacy & 1,000/1,000 & 0.4189 & 0.7909 & 100.00\% [99.62, 100.00] & conservative \\
ANES & Rep. thermometer & Strong MAR & Frame-AIPW & 995/1,000 & $-0.0254$ & 0.7303 & 96.10\% [94.71, 97.13] & compatible \\
ANES & Rep. thermometer & MNAR & Naive synthetic & 1,000/1,000 & 13.5457 & 13.5457 & 0.00\% [0.00, 0.38] & \textbf{failure} \\
ANES & Rep. thermometer & MNAR & H\'ajek IPW & 1,000/1,000 & 5.4403 & 5.4671 & 0.10\% [0.02, 0.56] & \textbf{failure} \\
ANES & Rep. thermometer & MNAR & Plain AIPW & 1,000/1,000 & 5.6989 & 5.7196 & 0.00\% [0.00, 0.38] & \textbf{failure} \\
ANES & Rep. thermometer & MNAR & PPI & 1,000/1,000 & 5.0051 & 5.0258 & 0.00\% [0.00, 0.38] & \textbf{failure} \\
ANES & Rep. thermometer & MNAR & PPI++ & 1,000/1,000 & 7.6333 & 7.6496 & 0.00\% [0.00, 0.38] & \textbf{failure} \\
ANES & Rep. thermometer & MNAR & Legacy & 1,000/1,000 & 10.8515 & 10.8687 & 0.00\% [0.00, 0.38] & \textbf{failure} \\
ANES & Rep. thermometer & MNAR & Frame-AIPW & 999/1,000 & 5.7040 & 5.7833 & 0.40\% [0.16, 1.02] & \textbf{failure} \\
Twin-2K & Purchase rate & Mild MAR & Naive synthetic & 1,000/1,000 & $-0.0105$ & 0.0105 & 0.00\% [0.00, 0.38] & \textbf{failure} \\
Twin-2K & Purchase rate & Mild MAR & H\'ajek IPW & 1,000/1,000 & $-0.0003$ & 0.0046 & 100.00\% [99.62, 100.00] & conservative \\
Twin-2K & Purchase rate & Mild MAR & Plain AIPW & 1,000/1,000 & 0.0001 & 0.0017 & 94.70\% [93.13, 95.93] & compatible \\
Twin-2K & Purchase rate & Mild MAR & PPI & 1,000/1,000 & 0.0001 & 0.0017 & 100.00\% [99.62, 100.00] & conservative \\
Twin-2K & Purchase rate & Mild MAR & PPI++ & 1,000/1,000 & 0.0001 & 0.0020 & 100.00\% [99.62, 100.00] & conservative \\
Twin-2K & Purchase rate & Mild MAR & Legacy & 1,000/1,000 & 0.0004 & 0.0035 & 100.00\% [99.62, 100.00] & conservative \\
Twin-2K & Purchase rate & Mild MAR & Frame-AIPW & 1,000/1,000 & 0.0001 & 0.0018 & 95.40\% [93.92, 96.53] & compatible \\
Twin-2K & Purchase rate & Strong MAR & Naive synthetic & 1,000/1,000 & $-0.0105$ & 0.0105 & 0.00\% [0.00, 0.38] & \textbf{failure} \\
Twin-2K & Purchase rate & Strong MAR & H\'ajek IPW & 1,000/1,000 & $-0.0006$ & 0.0093 & 99.60\% [98.98, 99.84] & conservative \\
Twin-2K & Purchase rate & Strong MAR & Plain AIPW & 1,000/1,000 & 0.0001 & 0.0027 & 95.70\% [94.26, 96.79] & compatible \\
Twin-2K & Purchase rate & Strong MAR & PPI & 1,000/1,000 & 0.0001 & 0.0026 & 100.00\% [99.62, 100.00] & conservative \\
Twin-2K & Purchase rate & Strong MAR & PPI++ & 1,000/1,000 & $-0.0009$ & 0.0028 & 99.80\% [99.27, 99.95] & conservative \\
Twin-2K & Purchase rate & Strong MAR & Legacy & 1,000/1,000 & $-0.0018$ & 0.0052 & 100.00\% [99.62, 100.00] & conservative \\
Twin-2K & Purchase rate & Strong MAR & Frame-AIPW & 999/1,000 & 0.0004 & 0.0050 & 96.70\% [95.40, 97.64] & conservative \\
Twin-2K & Purchase rate & MNAR & Naive synthetic & 1,000/1,000 & $-0.0105$ & 0.0105 & 0.00\% [0.00, 0.38] & \textbf{failure} \\
Twin-2K & Purchase rate & MNAR & H\'ajek IPW & 1,000/1,000 & 0.0120 & 0.0228 & 22.20\% [19.73, 24.88] & \textbf{failure} \\
Twin-2K & Purchase rate & MNAR & Plain AIPW & 1,000/1,000 & 0.0178 & 0.0179 & 0.00\% [0.00, 0.38] & \textbf{failure} \\
Twin-2K & Purchase rate & MNAR & PPI & 1,000/1,000 & 0.0107 & 0.0109 & 29.40\% [26.66, 32.30] & \textbf{failure} \\
Twin-2K & Purchase rate & MNAR & PPI++ & 1,000/1,000 & 0.0336 & 0.0337 & 0.00\% [0.00, 0.38] & \textbf{failure} \\
Twin-2K & Purchase rate & MNAR & Legacy & 1,000/1,000 & 0.0671 & 0.0672 & 0.00\% [0.00, 0.38] & \textbf{failure} \\
Twin-2K & Purchase rate & MNAR & Frame-AIPW & 928/1,000 & 0.0180 & 0.0190 & 4.20\% [3.12, 5.63] & \textbf{failure} \\
Twin-2K & Behavioral-econ. score & Mild MAR & Naive synthetic & 1,000/1,000 & $-0.0291$ & 0.0291 & 0.00\% [0.00, 0.38] & \textbf{failure} \\
Twin-2K & Behavioral-econ. score & Mild MAR & H\'ajek IPW & 1,000/1,000 & $-0.0007$ & 0.0072 & 100.00\% [99.62, 100.00] & conservative \\
Twin-2K & Behavioral-econ. score & Mild MAR & Plain AIPW & 1,000/1,000 & $-0.0001$ & 0.0035 & 92.90\% [91.14, 94.33] & \textbf{failure} \\
Twin-2K & Behavioral-econ. score & Mild MAR & PPI & 1,000/1,000 & $-0.0002$ & 0.0037 & 100.00\% [99.62, 100.00] & conservative \\
Twin-2K & Behavioral-econ. score & Mild MAR & PPI++ & 1,000/1,000 & $-0.0002$ & 0.0040 & 100.00\% [99.62, 100.00] & conservative \\
Twin-2K & Behavioral-econ. score & Mild MAR & Legacy & 1,000/1,000 & 0.0001 & 0.0060 & 100.00\% [99.62, 100.00] & conservative \\
Twin-2K & Behavioral-econ. score & Mild MAR & Frame-AIPW & 1,000/1,000 & $-0.0002$ & 0.0052 & 94.70\% [93.13, 95.93] & compatible \\
Twin-2K & Behavioral-econ. score & Strong MAR & Naive synthetic & 1,000/1,000 & $-0.0291$ & 0.0291 & 0.00\% [0.00, 0.38] & \textbf{failure} \\
Twin-2K & Behavioral-econ. score & Strong MAR & H\'ajek IPW & 1,000/1,000 & $-0.0006$ & 0.0097 & 100.00\% [99.62, 100.00] & conservative \\
Twin-2K & Behavioral-econ. score & Strong MAR & Plain AIPW & 1,000/1,000 & 0.0000 & 0.0054 & 94.10\% [92.46, 95.40] & compatible \\
Twin-2K & Behavioral-econ. score & Strong MAR & PPI & 1,000/1,000 & 0.0014 & 0.0054 & 100.00\% [99.62, 100.00] & conservative \\
Twin-2K & Behavioral-econ. score & Strong MAR & PPI++ & 1,000/1,000 & 0.0021 & 0.0059 & 98.90\% [98.04, 99.38] & conservative \\
Twin-2K & Behavioral-econ. score & Strong MAR & Legacy & 1,000/1,000 & 0.0036 & 0.0094 & 100.00\% [99.62, 100.00] & conservative \\
Twin-2K & Behavioral-econ. score & Strong MAR & Frame-AIPW & 999/1,000 & $-0.0006$ & 0.0101 & 95.20\% [93.69, 96.36] & compatible \\
Twin-2K & Behavioral-econ. score & MNAR & Naive synthetic & 1,000/1,000 & $-0.0291$ & 0.0291 & 0.00\% [0.00, 0.38] & \textbf{failure} \\
Twin-2K & Behavioral-econ. score & MNAR & H\'ajek IPW & 1,000/1,000 & 0.0417 & 0.0436 & 6.30\% [4.95, 7.98] & \textbf{failure} \\
Twin-2K & Behavioral-econ. score & MNAR & Plain AIPW & 1,000/1,000 & 0.0442 & 0.0444 & 0.00\% [0.00, 0.38] & \textbf{failure} \\
Twin-2K & Behavioral-econ. score & MNAR & PPI & 1,000/1,000 & 0.0229 & 0.0233 & 11.90\% [10.04, 14.05] & \textbf{failure} \\
Twin-2K & Behavioral-econ. score & MNAR & PPI++ & 1,000/1,000 & 0.0661 & 0.0663 & 0.00\% [0.00, 0.38] & \textbf{failure} \\
Twin-2K & Behavioral-econ. score & MNAR & Legacy & 1,000/1,000 & 0.1113 & 0.1115 & 0.00\% [0.00, 0.38] & \textbf{failure} \\
Twin-2K & Behavioral-econ. score & MNAR & Frame-AIPW & 989/1,000 & 0.0440 & 0.0453 & 1.40\% [0.84, 2.34] & \textbf{failure} \\
\end{longtable}
\normalsize
\endgroup

\section{Extended ANES Results}
\label{app:anes_extended}

This appendix presents extended results from the ANES analysis.

\subsection{Full Subgroup-Level Diagnostics}

Table~\ref{tab:anes_age_calibration} reports conditional calibration slopes by age group rather than party ID. The contrast is striking: while party-level calibration reveals sign flips and slopes far from 1, age-level calibration shows reasonable performance ($\beta_g \in [0.61, 0.98]$), with Falcon-40B more accurately calibrated for older respondents.

\begin{table}[htbp]
\centering
\caption{ANES conditional calibration: age subgroup slopes for Falcon-40B thermometers.}
\label{tab:anes_age_calibration}
\small
\begin{tabular}{l c c c c c}
\toprule
\textbf{Thermometer} & \textbf{18--29} & \textbf{30--44} & \textbf{45--59} & \textbf{60+} & \textbf{Decision} \\
\midrule
Dem.\ Party & 0.796 & 0.800 & 0.832 & 0.980 & Proceed \\
Rep.\ Party & 0.611 & 0.609 & 0.667 & 0.756 & Proceed \\
\bottomrule
\end{tabular}
\end{table}

This result has a practical implication: the framework's decision depends on which covariates define the subgroups for Diagnostic~2. When subgroups are defined by the variable most relevant to concept shift (party ID, which directly affects thermometer responses), the diagnostic correctly flags the problem. When subgroups are defined by a less relevant variable (age), the diagnostic gives a false ``proceed'' signal. This motivates the practitioner recommendation to define calibration subgroups along the dimensions most likely to exhibit heterogeneous concept shift.

\subsection{Bias Decomposition by Thermometer}

\begin{table}[htbp]
\centering
\caption{ANES bias decomposition across analysis settings.}
\label{tab:anes_decomposition}
\small
\begin{tabular}{l c c c c c}
\toprule
\textbf{Analysis} & $\hat{\Delta}_{\text{total}}$ & $\hat{\Delta}_{\text{cov}}$ & $\hat{\Delta}_{\text{con}}$ & \textbf{Cov \%} & \textbf{Con \%} \\
\midrule
Full sample, Dem.\ Party & $+5.92$ & $-0.20$ & $+6.12$ & 3.2\% & 96.8\% \\
Full sample, Rep.\ Party & $+6.20$ & $+0.17$ & $+6.03$ & 2.7\% & 97.3\% \\
Young Dems, Dem.\ Party & $+33.06$ & $+11.49$ & $+21.57$ & 34.8\% & 65.2\% \\
Older Reps, Rep.\ Party & $+28.81$ & $+14.26$ & $+14.55$ & 49.5\% & 50.5\% \\
\bottomrule
\end{tabular}
\end{table}

The full-sample analyses confirm near-zero covariate shift ($< 3\%$), consistent with the matched-demographics design in \citet{bisbee2024synthetic}. The subsampled analyses show meaningful covariate shift (35--50\%) as expected when restricting the target to demographic subsets.

\begingroup
\subsection{Outcome Model Comparison}
\endgroup

\begin{table}[htbp]
\centering
\caption{ANES estimator comparison: logistic propensity with linear vs.\ random forest outcome models.}
\label{tab:anes_methods}
\small
\begin{tabular}{l l c c c}
\toprule
\textbf{Thermometer} & \textbf{Outcome Model} & $|\textbf{Bias}_{\text{DR}}|$ & \textbf{Bias Red.} & \textbf{CI} \\
\midrule
Dem.\ Party & Linear & 0.029 & 99.5\% & $\checkmark$ \\
Dem.\ Party & Random forest & 0.156 & 97.4\% & $\checkmark$ \\
Rep.\ Party & Linear & 0.022 & 99.6\% & $\checkmark$ \\
Rep.\ Party & Random forest & 0.263 & 95.8\% & $\checkmark$ \\
\bottomrule
\end{tabular}
\end{table}

Linear outcome models outperform random forest on this data, likely because the full-sample analysis has near-perfect covariate overlap (ESS/N $= 0.95$) and the concept shift is approximately linear in demographics. Random forest introduces unnecessary variance through overfitting. Both methods achieve $> 95\%$ bias reduction and valid CI coverage.

\subsection{Three-Model Cross-LLM Stability}

 We extend the cross-LLM stability analysis to three models: GPT-3.5, GPT-4, and Falcon-40B. At the party-subgroup level (Democrat, Independent, and Republican), every pairwise Spearman correlation equals 1.0 for both the Democratic Party and Republican Party thermometers.

\begin{table}[htbp]
\centering
\caption{Cross-LLM subgroup estimates: feeling thermometer toward Republican Party.}
\label{tab:anes_cross_llm_rep}
\small
\begin{tabular}{l c c c c}
\toprule
\textbf{Subgroup} & \textbf{Human (ANES)} & \textbf{GPT-3.5} & \textbf{GPT-4} & \textbf{Falcon-40B} \\
\midrule
Democrat       & 24.0 & 19.3 & 36.2 & 25.0 \\
Independent    & 40.6 & 46.5 & 53.2 & 32.7 \\
Republican     & 66.1 & 86.2 & 81.9 & 54.4 \\
\bottomrule
\end{tabular}
\end{table}

\noindent Despite identical rankings, the models exhibit markedly different concept shift patterns. GPT-3.5 and GPT-4 show positive bias for Republicans (overestimating by 15--20~pp) but opposite biases for Democrats (GPT-3.5 underestimates by 5~pp, GPT-4 overestimates by 12~pp). Falcon-40B underestimates uniformly across all subgroups. These heterogeneous biases underscore the importance of model-specific calibration.

\section{Cross-Fitting Implementation Details}
\label{app:crossfitting}

This appendix provides implementation details for the $K$-fold cross-fitting procedure.

\subsection{Nuisance Estimator Choices}

\paragraph{Outcome model $\hat{g}$.} The concept shift function $g(x) = \muT(x) - \muS(x)$ can be estimated by:
\begin{enumerate}[label=(\alph*)]
  \item \emph{Difference of regressions:} Fit $\hat{\mu}_{\cT}$ on the real sample and $\hat{\mu}_{\cS}$ on the synthetic sample separately; take $\hat{g}(x) = \hat{\mu}_{\cT}(x) - \hat{\mu}_{\cS}(x)$.
  \item \emph{Direct regression:} Pool data with source indicator $D$, fit $\hat{g}(x) = \what{\E}[Y \mid X = x, D = \cT] - \what{\E}[Y \mid X = x, D = \cS]$.
  \item \emph{Parametric:} For small $n$, use $\hat{g}(x) = \alpha + \beta^\top x$ (linear calibration).
\end{enumerate}

We use random forest for the rich CRM scenario, linear regression for survey-only, and skip outcome calibration for aggregate benchmarks.

\paragraph{Propensity model $\hat{w}$.} The density ratio $w(x) = p_{\cT}(x) / p_{\cS}(x)$ is estimated via a domain classifier:
\begin{equation}
  \hat{w}(x) = \frac{1 - \hat{\pi}(x)}{\hat{\pi}(x)} \cdot \frac{N}{n},
  \label{eq:propensity_model}
\end{equation}
where $\hat{\pi}(x) = \what{P}(\text{synthetic} \mid X = x)$ is fitted by logistic regression or random forest on the pooled sample.

\subsection{Weight Trimming}

Extreme propensity weights can destabilize the DR estimator. We trim at the 1st and 99th percentiles:
\begin{equation}
  \tilde{w}_i = \min\bigl(\max(w_i, w_{0.01}), w_{0.99}\bigr),
  \label{eq:trimming}
\end{equation}
where $w_{q}$ denotes the $q$-th quantile of the empirical weight distribution. This introduces a small bias but substantially reduces variance.

\begingroup
\subsection{Outcome and Propensity Model Checks}

The reported ANES analysis holds the logistic propensity model fixed and compares linear and random-forest outcome models. The linear model has lower absolute DR bias for both thermometers (0.029 versus 0.156 for the Democratic Party; 0.022 versus 0.263 for the Republican Party). Both models achieve valid CI coverage.

The frozen DGP battery does not compare logistic and random-forest propensity estimators. Instead, it crosses correctly and incorrectly specified logistic propensity models with correctly and incorrectly specified linear outcome models. These results therefore do not support a ranking between logistic and random-forest propensity estimation.

\textbf{Implementation.} The frozen battery uses logistic propensity models. Entropy balancing \citep{hainmueller2012entropy} remains a possible alternative that targets covariate balance directly, but it is not implemented in the current codebase.
\endgroup


\begingroup
\section{Worked Example: D2C Sustainable Products}
\label{app:worked_example_revised}

This is a hypothetical workflow; the values illustrate routing rather than empirical performance. Steps~1--3 diagnose overlap, cross-model stability, and conditional calibration as in the main-text example.

\paragraph{Step 4: Invitation-frame check.}
Suppose the company's CRM records every customer invited to the survey, the covariates and stored LLM prediction for each invited customer, who responded, and the human outcome for respondents. If MAR is credible given the recorded variables and every cross-fitted propensity remains above the prespecified floor, frame-AIPW is available for the segments that pass the diagnostic checks. The 55+ rural segment still routes to Walk Away because its conditional-calibration result falls outside the operational range.

Suppose instead that the 200 responses came from a stand-alone panel with no record of who was invited but did not respond. The first three diagnostics are unchanged, but correction is unavailable because the response process cannot be learned from respondents alone. The company must use a design-valid human estimator when one exists, reconstruct the invitation frame, collect better human data, or Walk Away.
\endgroup


\begingroup
\section{Extensions Not Evaluated}
\label{app:consumer_extensions}

This paper implements and evaluates a scalar weighted frame mean.

\paragraph{Product rankings.}
Ranking $K$ concepts requires a vector estimand, joint covariance estimation, and a multiplicity procedure for pairwise or simultaneous claims.

\paragraph{Segment treatment effects.}
Segment effects require a well-defined treatment, potential outcomes, and segment-specific frame support, MAR, and positivity. They are not obtained by relabeling the scalar frame mean as a CATE.

\paragraph{Willingness-to-pay and demand curves.}
These targets require price support, a functional estimand, and explicit shape restrictions. Correcting one mean does not establish a corrected curve or preserve its shape.

 These extensions require methods and evidence beyond the frame-AIPW implementation evaluated here.
\endgroup


\clearpage
\begingroup
\section{Data Requirements and Actions}
\label{app:data_requirements}

\begin{table}[htbp]

\centering
\caption{Correction availability depends on the observed response frame, not a fixed respondent count.}
\label{tab:data_requirements_appendix}
\small
\begin{tabularx}{\textwidth}{>{\raggedright\arraybackslash}p{4.2cm} >{\raggedright\arraybackslash}p{3.2cm} X}
\toprule
\textbf{Available data} & \textbf{Correction status} & \textbf{Action} \\
\midrule
Complete invitation frame: $W_i$, $f_i$, $a_i$, and $R_i$ for all invited units; $Y_i$ for respondents & Eligible for frame-AIPW only after credible MAR and positivity checks & Fit cross-fitted nuisance models, apply the prespecified propensity floor, and report the result under the declared finite-frame or superpopulation reading. \\[6pt]
Respondent-only survey, with no record of invited nonrespondents & Correction unavailable & Use a design-valid human estimator if the sampling design supports the target; otherwise reconstruct the frame or collect better human data. \\[6pt]
Aggregate population benchmarks only & Correction unavailable & Use the benchmarks for marginal diagnostics or sensitivity analysis, not frame-AIPW correction. \\[6pt]
Invitation frame exists, but MAR is not credible or positivity fails & Point correction unsupported & Report bounds or sensitivity analysis where possible, collect better human data, or Walk Away. \\
\bottomrule
\end{tabularx}
\end{table}
\FloatBarrier
\endgroup


\begingroup
\section{Code Documentation and Reproduction}
\label{app:code_revised}

The Python package requires Python~3.9 or later. Install the package and development dependencies with:

\begin{center}
\texttt{python -m pip install -e '.[dev]'}
\end{center}

\paragraph{Primary manuscript interfaces.}
The package root exports the three diagnostics, \texttt{bias\_decomposition}, \texttt{frame\_aipw}, and \texttt{run\_decision\_framework}. The secondary \texttt{frame\_aipw\_stacked} estimator is implemented but excluded from the paper-number tables. The historical \texttt{dr\_correction} function is not exported from the package root; it is retained only through \texttt{synthetic\_users.legacy\_correction} with a deprecation warning.

\paragraph{Verification.}
Run the package tests with:

\begin{center}
\texttt{env -u PYTHONPATH python -m pytest tests/ -q}
\end{center}

\paragraph{Frozen final experiments.}

Run the four commands frozen in \path{experiments/final_simulation_matrix.json}. The first two execute the identified and MNAR expected-failure batteries; the other two apply controlled masking to ANES and Twin-2K. The combined outputs are:

\begin{itemize}[leftmargin=2em]
  \item \texttt{notes/artifacts/battery\_final.csv.gz}
  \item \texttt{notes/artifacts/semi\_synthetic\_final.csv.gz}
\end{itemize}

The frozen matrix records seeds, cell counts, replication counts, estimator sets, acceptance rules, and exact commands. Runtime depends on hardware and parallelism; we report no fixed laptop timing.
\endgroup

\section{Formal DISDE Decomposition and Assumptions}
\label{app:formal_decomposition}

This appendix provides the formal telescoping decomposition and regularity assumptions summarized in Section~3 of the main text.

\subsection{Formal Decomposition via DISDE}

Following the DISDE telescoping sum of \citet[Equation 2.1]{cai2025disde}, we introduce a shared intermediate distribution to separate covariate and concept shift cleanly.

\begin{definition}[Shared Distribution]
\label{app:def:shared_distribution}
Define the shared distribution $\Sx$ with density
\begin{equation}
  p_{\Sx}(x) \propto \frac{p_{\cS}(x) \cdot p_{\cT}(x)}{p_{\cS}(x) + p_{\cT}(x)},
  \label{app:eq:shared_density}
\end{equation}
which places mass only where both $\pS$ and $\pT$ have support.
\end{definition}

Define the conditional risk under each population:
\begin{equation}
  R_{\cS}(x) \coloneqq \muS(x) = \E_{\cS}[Y \mid X = x], \qquad
  R_{\cT}(x) \coloneqq \muT(x) = \E_{\cT}[Y \mid X = x].
  \label{app:eq:conditional_risk}
\end{equation}

The DISDE telescoping decomposition gives:
\begin{align}
  \tauT - \tauS &= \bigl\{\E_{\Sx}[R_{\cS}(X)] - \E_{\pS}[R_{\cS}(X)]\bigr\} \label{app:eq:disde_term1} \\
  &\quad + \bigl\{\E_{\Sx}[R_{\cT}(X) - R_{\cS}(X)]\bigr\} \label{app:eq:disde_term2} \\
  &\quad + \bigl\{\E_{\pT}[R_{\cT}(X)] - \E_{\Sx}[R_{\cT}(X)]\bigr\}. \label{app:eq:disde_term3}
\end{align}

\noindent Terms \eqref{app:eq:disde_term1} and \eqref{app:eq:disde_term3} capture covariate shift (the change in marginal $X$ distribution from $\pS$ to $\Sx$ and from $\Sx$ to $\pT$, respectively), while term \eqref{app:eq:disde_term2} isolates concept shift (the change in conditional response $R_{\cT}(x) - R_{\cS}(x)$, evaluated under the shared distribution where both populations have adequate support).

\begin{remark}[Interpretation for Practitioners]
Term \eqref{app:eq:disde_term2} measures how differently the LLM ``thinks'' compared to real consumers \emph{at the same demographics}. If this term dominates, reweighting the LLM panel to match the target demographics will not fix the problem. The practitioner needs outcome model calibration (our DR correction) or must walk away.
\end{remark}

\subsection{Assumptions}

\begin{assumption}[Shared Support]
\label{app:ass:shared_support}
For all $x \in \cX$, if $p_{\cT}(x) > 0$, then $p_{\cS}(x) > 0$. That is, the LLM's implicit population covers the target population's demographic profiles: $\mathrm{supp}(\pT) \subseteq \mathrm{supp}(\pS)$.
\end{assumption}

\begin{assumption}[Bounded Density Ratios]
\label{app:ass:bounded_ratios}
There exists a constant $M < \infty$ such that
\begin{equation}
  \sup_{x \in \cX} \frac{p_{\cT}(x)}{p_{\cS}(x)} \leq M.
  \label{app:eq:bounded_density}
\end{equation}
\end{assumption}

\begin{assumption}[Regularity]
\label{app:ass:regularity}
$\E_{\cS}[Y^2] < \infty$ and $\E_{\cT}[Y^2] < \infty$. The conditional response functions $\muS(\cdot)$ and $\muT(\cdot)$ are measurable and square-integrable under their respective distributions.
\end{assumption}

\begingroup
\begin{assumption}[Generation and Dependence Disclosure]
\label{app:ass:generation_dependence}
Stored predictions must be frozen before human-outcome and response evaluation, available for every frame unit, and free of target-$Y$ or response-label leakage. Shared prompts, conversational memory, or batch-level context can induce dependence across units; related units must remain in the same cross-fitting fold. Cluster-level inference is required when the declared sampling or response design is clustered. Deterministic generation is allowed because frame-AIPW treats $f$ as a fixed auxiliary variable and does not require synthetic randomness. Real observations follow the declared finite-frame response design, or an explicitly labeled iid or iid-cluster superpopulation design.
\end{assumption}

The overlap and decomposition diagnostics can be computed descriptively without interpreting synthetic rows as independent respondents. Sampling-based uncertainty or superpopulation generalization requires a declared sampling law and group- or cluster-aware analysis when units share generation context.

\paragraph{Estimation.} The decomposition is estimated via a domain classifier that pools synthetic and real samples, trains $\hat{\pi}(x) = \what{P}(\text{synthetic} \mid X = x)$, and converts it to importance weights $w_i \propto (1 - \hat{\pi}(X_i^{\cS})) / \hat{\pi}(X_i^{\cS})$. Shared-distribution expectations are computed as importance-weighted averages using \eqref{app:eq:shared_density}. The amount of human data needed depends on the target, design, overlap, outcome variability, and diagnostic precision; no universal respondent count is established.
\endgroup


\begingroup
\section{Cross-Dataset Diagnostic Comparison}
\label{app:cross_dataset_revised}

These paired datasets support descriptive diagnostic comparisons. They contain no real invitation frame or observed response mechanism, so this table does not report correction performance.

\begin{table}[htbp]

\centering
\small
\caption{Cross-dataset diagnostic comparison and scoped action.}
\label{app:tab:cross_dataset_revised}
\begin{tabularx}{\textwidth}{>{\raggedright\arraybackslash}p{3.2cm} >{\raggedright\arraybackslash}p{5.1cm} X}
\toprule
\textbf{Dataset or target} & \textbf{Diagnostic evidence} & \textbf{Scoped action} \\
\midrule
ANES full sample & ESS/N $=0.954$; subgroup slopes $[-0.004,0.274]$ with a sign reversal & Walk Away for affected outcomes. \\[5pt]
ANES young Democrats & ESS/N $=0.110$; no within-subgroup sign flip & Outside Trust; correction unavailable without an invitation frame. \\[5pt]
ANES older Republicans & ESS/N $=0.260$; no within-subgroup sign flip & Outside Trust; correction unavailable without an invitation frame. \\[5pt]
Twin-2K full sample & ESS/N $=1.000$; subgroup slopes $[0.80,0.93]$ & No tested diagnostic failure for this task, model, and sample. \\[5pt]
Twin-2K high-income target & ESS/N $=0.401$; subgroup slopes $[0.21,0.28]$ & Outside Trust; correction unavailable without an invitation frame. \\[5pt]
Twin-2K ages 18--29 & ESS/N $=0.203$; subgroup slopes $[0.16,0.54]$ & Outside Trust; correction unavailable without an invitation frame. \\
\bottomrule
\end{tabularx}
\end{table}

The comparison shows that diagnostic behavior is task- and subgroup-specific. The same Twin-2K system has no flagged full-sample pricing failure but falls outside the operational Trust range for targeted subgroups and cognitive-bias tasks. ANES also shows that subgroup choice matters: party-defined groups reveal failures that age-defined groups can obscure. These findings determine where additional human evidence is needed; they do not certify MAR, positivity, or correction availability.
\endgroup


\begingroup
\section{Sample Size: No Universal Threshold}
\label{app:sample_size_revised}

The frozen final experiments do not vary respondent count while holding the frame, response design, nuisance complexity, and estimand fixed. They therefore do not establish a minimum $n$.

The amount of human data required depends on the response rate, frame size, weight concentration, overlap, outcome variability, nuisance complexity, subgroup, and target reading. Applications must evaluate held-out calibration, positivity, influence-score behavior, and uncertainty for their own design.
\endgroup


\begingroup
\small
\singlespacing
\setlength{\bibsep}{6pt plus 2pt minus 1pt}
\bibliography{references}
\endgroup

\end{document}